\documentclass[10pt,journal,compsoc]{IEEEtran}

\ifCLASSOPTIONcompsoc

\usepackage[nocompress]{cite}
\else
\usepackage{cite}
\fi

\ifCLASSINFOpdf

\else

\fi
\usepackage{amsmath,lipsum}
\usepackage{cuted}
\usepackage{graphicx}
\usepackage{epstopdf}
\usepackage{amsmath}
\usepackage{subeqnarray}
\usepackage{cases}
\usepackage{algorithmic}
\usepackage{algorithm}
\usepackage{array}
\usepackage{multicol}
\usepackage{stfloats}

\usepackage{url}
\usepackage{harpoon}
\usepackage{booktabs}
\usepackage{multirow}
\usepackage{amsfonts}
\usepackage{xcolor}
\usepackage{tabularx}
\usepackage{amssymb,varwidth}
\usepackage{bm}
\usepackage{float}
\usepackage{graphicx}
\usepackage{subfigure}
\usepackage{ntheorem}
\usepackage{mathtools}
\allowdisplaybreaks[3]
\newenvironment{proof}{{\indent \it Proof:\quad}}{\hfill $\blacksquare$\par}
\theoremseparator{:}

\newtheorem{theorem}{Theorem}
\newtheorem{assumption}{\textbf{Assumption}}
\newtheorem{lemma}{\textbf{Lemma}}
\newtheorem{remark}{\textbf{Remark}}
\begin{document}
\begin{sloppypar}

\bibliographystyle{IEEEtran}

\title{Lightweight Federated Learning over Wireless Edge Networks}
\author{Xiangwang Hou, \emph{Graduate Student Member, IEEE,}
        Jingjing Wang, \emph{Senior Member, IEEE},  \\ Jun Du, \emph{Senior Member, IEEE}, Chunxiao Jiang, \emph{Fellow, IEEE},   \\  Yong Ren, \emph{Senior Member, IEEE}, Dusit Niyato, \emph{Fellow, IEEE}
\thanks{X. Hou, J. Du, Y. Ren are with the Department of Electronic Engineering, Tsinghua University, Beijing, 100084, China. (E-mail: hxw21@mails.tsinghua.edu.cn, \{jundu, reny\}@tsinghua.edu.cn.)}
\thanks{J. Wang is with the School of Cyber Science and Technology, Beihang University, Beijing 100191, China. (Email: drwangjj@buaa.edu.cn)}
\thanks{C. Jiang is with the Tsinghua Space Center, Tsinghua University, Beijing, 100084, China. (E-mail: jchx@tsinghua.edu.cn).}
\thanks{D. Niyato is with the College of Computing and Data Science, Nanyang Technological University, Singapore 639798. (E-mail: dniyato@ntu.edu.sg)}
}

\markboth{}
{Hou \MakeLowercase{\textit{et al.}}: Lightweight Federated Learning over Wireless Edge Networks}

\IEEEtitleabstractindextext{

\begin{abstract}
With the exponential growth of smart devices connected to wireless networks, data production is increasing rapidly, requiring machine learning (ML) techniques to unlock its value. However, the centralized ML paradigm raises concerns over communication overhead and privacy. Federated learning (FL) offers an alternative at the network edge, but practical deployment in wireless networks remains challenging. This paper proposes a lightweight FL (LTFL) framework integrating wireless transmission power control, model pruning, and gradient quantization. We derive a closed-form expression of the FL convergence gap, considering transmission error, model pruning error, and gradient quantization error. Based on these insights, we formulate an optimization problem to minimize the convergence gap while meeting delay and energy constraints. To solve the non-convex problem efficiently, we derive closed-form solutions for the optimal model pruning ratio and gradient quantization level, and employ Bayesian optimization for transmission power control. Extensive experiments on real-world datasets show that LTFL outperforms state-of-the-art schemes.
\end{abstract}

}
\maketitle
\IEEEdisplaynontitleabstractindextext
\IEEEpeerreviewmaketitle

\section{Introduction}
Over the past few years, the cloud-based artificial intelligence (AI)  has achieved remarkable success, with devices continuously transmitting their collected data to cloud servers for centralized machine learning (ML) training. Recently, the proliferation of intelligent devices (e.g.,  smart phones, smart cars), results in a the phenomenal growth of the data generated at the edge of the networks, which led to a surging demand for AI services. However, the conventional cloud-based AI approach's substantial consumption of communication resources and potential privacy risks have increasingly made it unable to meet the growing service demands \cite{9060868}.

Fortunately, as the computational capabilities of the access points (APs) and devices continue to advance, on-demand computing services in the network edge can be performed \cite{8736011,10305704}.  Benefiting from this, implementing federated learning (FL) over the network edge and devices for training ML models becomes a more optimal choice compared with the cloud-based AI.  FL \cite{mcmahan2017communication} is a novel distributed ML paradigm that enables an edge server to coordinate massive devices to collaboratively train a shared ML model through interacting few parameters instead of the huge raw data, which endows it with inherent advantages in saving communication resources and protecting privacy.

{Despite the promising future of performing FL in the network edge, there still exist significant challenges when deploy FL in realistic wireless networks \cite{10038617}, \cite{9809924},\cite{10597905}.
\begin{itemize}
  \item  \emph{Unreliable wireless links}. Wireless channels are unreliable, resulting in occasional failures in parameter exchange between the edge server and devices.
  \item  \emph{Heavy communication overhead}. The communication overhead remains significant due to the increasing complexity of ML models, even though only parameter exchanges are required.
  \item \emph{Limited computational resources}.  Due to limited resources, training ML models with even millions of parameters on these resource-constrained devices can lead to unbearable delay and even breakdown.
\end{itemize}

Hence, there is an imminent necessity to devise a lightweight FL (LTFL) framework that adapts to the wireless edge networks.
{In this context, "lightweight" refers to an FL framework that efficiently reduces both computational and communication overhead while maintaining model performance. }}

\subsection{Related Works}
Some prior studies have endeavored to address the issues mentioned above separately.  Some early efforts have explored how wireless transmission affects the FL's convergence process.  Chen \emph{et al.}\cite{9210812} presented the theoretical analysis of the impact of both wireless transmission error and device selection on the convergence error and proposed a joint learning, wireless resource allocation, and device selection strategy for minimizing the training loss with time and energy consumption constraints.  Aiming to support FL in cell-free massive multiple-input multiple-output (MIMO) networks, Vu \emph{et al.} \cite{9124715} formulated an optimization problem to jointly schedule the local accuracy, transmission power, data rate, and user processing frequency.  Moreover,  considering the dynamic-changing network conditions, Xu \emph{et al.} \cite{9237168} presented a long-term device selection and resource allocation strategy for wireless FL.

Some works investigated communication compression to reduce the communication overhead. Li \emph{et al.} \cite{9488839} introduced top-$k$ sparsification to shorten the transmitted bits and proposed an energy-efficiency compression control algorithm to  flexibly decide the gradient sparsity and global update frequency. Moreover, Reisizadeh \emph{ et al.} \cite{8786146} designed a decentralized gradient descent (GD)  algorithm with quantization, where the participants execute local updates by considering the quantized information from their nearby devices. Our previous work in \cite{10368103} introduced a soft actor-critic-based approach for dynamic device selection, gradient quantization, and resource allocation to facilitate FL in dynamically changing network environments. Fan \emph{et al.} \cite{9912341} exploited the waveform superposition property of wireless channels and presented a 1-bit compression sensing scheme to reduce communication overhead. Furthermore, Sattler \emph{et al.} \cite{STC} proposed sparse ternary compression (STC), which extends top-$k$ gradient sparsification with downstream compression, ternarization, and Golomb encoding to further enhance communication efficiency in federated learning.

Recently, some studies exploited model pruning techniques to reduce the computational complexity of local training, thereby facilitating the deployment of FL on resource-constrained devices. Lin \emph{et al.} \cite{lin2020dynamic} introduced model pruning into FL and incorporated error feedback to improve the convergence performance. Taking the wireless transmission into consideration, Liu \emph{et al.} \cite{9598845} proposed to jointly schedule device selection, resource allocation, and model pruning to maximize the convergence rate with learning latency budget, while Jiang \emph{et al.} \cite{10050151} presented a multi-armed bandit based online algorithm to dynamically adjust the model pruning ratio to minimize the convergence time with accuracy guarantee. Furthermore, Prakash \emph{et al.} \cite{prakash2022iot} combined model pruning and gradient quantization to conceive an efficient FL, but ignored the impact of unreliable wireless transmission on FL's convergence.

\begin{table*}[t]
\centering
\caption{Comparison between the state-of-the-art literature and this paper.}\label{LiteratureComparison}
{
\begin{tabular}{cccccccccccccc}
\toprule
\textbf{Literature  }                  & \cite{9210812} & \cite{9124715} &   \cite{9237168} &\cite{9488839} & \cite{8786146} & \cite{10368103} &\cite{9912341} & \cite{STC} & \cite{lin2020dynamic} &  \cite{9598845} & \cite{10050151} &\cite{prakash2022iot} & \textbf{This paper} \\ \midrule
\midrule
\textbf{Wireless Resource Allocation} & \checkmark                                                                 & \checkmark                                                                                                                                & \checkmark                                                                & \checkmark                                                               &             &\checkmark                                                                      & \checkmark                                                                 &                         &                                                                           & \checkmark                                                                 & \checkmark                                                                  &                                                    & \checkmark          \\ \midrule
\textbf{Communication Compression }    &                                                                    &                                                                                                                                      &                                                                   & \checkmark                                                                 & \checkmark             &\checkmark                                                                 & \checkmark                                                                 & \checkmark    &                                                                       &                                                                    &                                                                       &\checkmark                                                                                & \checkmark \\ \midrule
\textbf{Model Pruning }                &                                                                    &                                                                                                                                   &                                                                   &                                                                   &                                &                                               &                                                                    &                        & \checkmark                                                                       & \checkmark                                                                 &\checkmark                                                                 & \checkmark                                                      & \checkmark         \\ \bottomrule
\end{tabular}
}
\end{table*}

\subsection{Motivations and Contributions}
{Essentially, FL entails multiple stages of computation and communication, making it inadequate to separately optimize the learning process from communication or computing perspectives. Hence, a lightweight design is required, where computational and communication efficiency are jointly improved. For instance, if a device has limited computational capability but favorable communication conditions, allocating more precise quantization bit widths and communication resources may yield marginal improvements in convergence speed. In contrast, improving the pruning  level of the ML model deployed on the aforementioned device would bring substantial enhancement of FL's convergence performance. Therefore,  a joint optimization from the aspects of computation and transmission is essential.

However, the existing literature lacks comprehensive investigation jointly considering unreliable wireless transmission, communication compression, and model pruning. To fulfill this gap, this paper proposes a LTFL framework encompassing wireless resource allocation, model pruning, and gradient quantization. TABLE \ref{LiteratureComparison} summarizes the comparison of the aforementioned literature and this paper.}

Specifically, the contributions are summarized as follows:

\begin{itemize}
  \item We propose a lightweight FL framework, LTFL, to accommodate resource-constrained devices and unreliable and scarce wireless communication resources. As far as we know, this is the first attempt to incorporate model pruning, quantization, and wireless resource allocation for wireless FL.
  \item We derive a closed-form expression of the upper bound of the impact of model pruning error, gradient quantization error and transmission error on the convergence process of FL. Guided by this findings, we jointly optimize the model pruning ratio, gradient quantization level, and transmission power control to reduce the training loss while satisfying the delay and energy consumption constraints.
   \item  Due to the formulated problem is non-convex, we decompose it into subproblems and design a two-stage algorithm to solve them. Firstly, we derive the closed-form optimal solutions of the pruning ratio and gradient quantization level. Then, we design a Bayesian optimization-based algorithm to obtain the optimal transmission power. Experimental results demonstrate that the proposed scheme performs better in reducing delay and saving energy than state-of-the-art schemes.
\end{itemize}

The rest of the paper is organized as follows. Section II reviews the state-of-the-art literature. Section III and Section IV present the system model and convergence analysis of FL, respectively. Section V demonstrates the performance evaluation and problem formulation, while Section VI introduces the algorithm. Section VII shows the experimental results, and Section VIII concludes the paper.

\section{System Model}
In the LTFL system, as depicted in Fig. \ref{LTFLarchitecture}, there exists an AP paired with an edge server and a set of $U$  devices denoted as $\mathcal{U}=\{1,2,\dots,u,\dots, U\}$ to perform FL. Each device $u$ has collected $N_u$ data samples represented by $\mathcal{D}_u =\{(\boldsymbol{x}_{u, 1}, \boldsymbol{y}_{u, 1}), (\boldsymbol{x}_{u, 2}, \boldsymbol{y}_{u, 2}),\dots, (\boldsymbol{x}_{u, N_u},\boldsymbol{y}_{u, N_u})\}$. The pair $\left(\boldsymbol{x}_{u, i}, \boldsymbol{y}_{u, i}\right)$ comprises a data sample $\boldsymbol{x}_{u, i} \in\mathbb{R}^{N_{in}\times 1}$ and its corresponding label $\boldsymbol{y}_{u, i}$, where $\mathbb{R}^{N_{in}\times1}$ indicates that the data sample consists of $N_{in}$ features. To reduce the computation cost, each device employs model pruning to simplify the model by eliminating unimportant computations. The devices then train their models in parallel and obtain their local gradients. Additionally,  aiming to reduce communication overhead, the devices use lower-precision representations to quantize their local gradients and  subsequently transmit them to the edge server. Noted that potential transmission errors may result in some updated gradients not being successfully received by the edge server. Once the edge server receives the quantized local gradients, it executes gradient aggregation and updates the global model accordingly. In the final, the edge server sends the newly updated model to the devices for subsequent iterations.

\begin{figure}
  \centering
  \includegraphics[width=9cm]{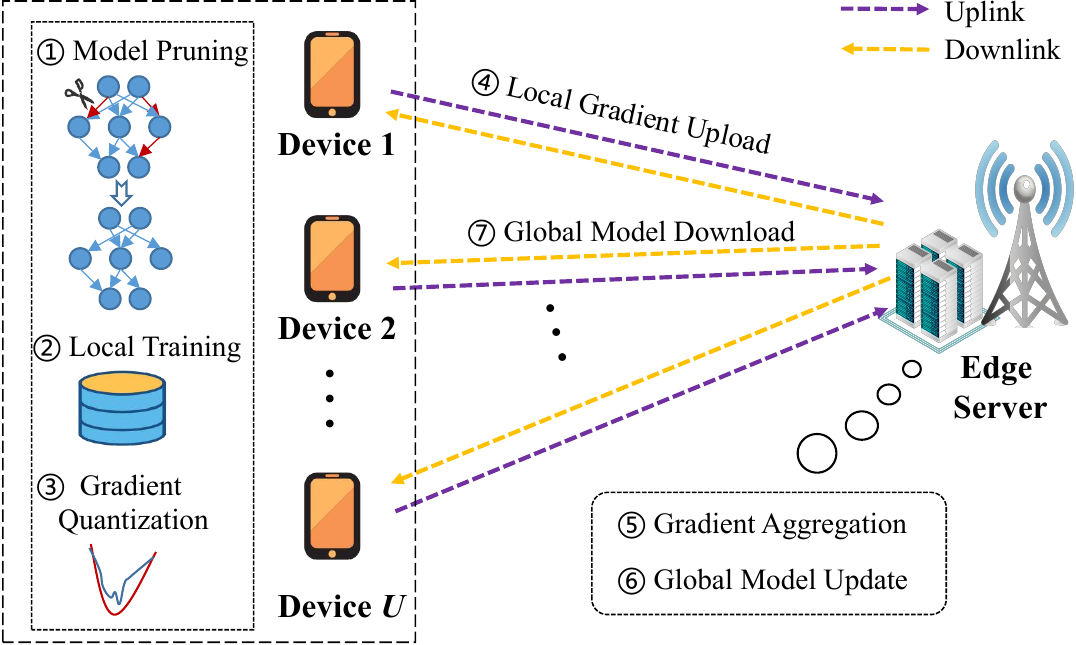}
  \caption{The architecture of LTFL over wireless networks.}  \label{LTFLarchitecture}
\end{figure}
\subsection{Transmission Model}
In the uplink transmission, the classical orthogonal frequency division multiple access (OFDM) is employed \cite{wang2024generative}. During the $n$-th iteration, the data rate from device $u$ to the edge server is given by  \cite{9292468}:
\begin{equation}\label{datarate}
R_u^{\mathrm{n}}\left( p^{n}_u\right)= B_u^{\mathrm{UL}} \mathbb{E}_{h^{n}_u}\left(\log _2\left(1+\frac{p^{n}_u h^{n}_u}{I^{n}_u+B_u^{\mathrm{UL}} N_0}\right)\right),
\end{equation}
where $\mathbb{E}_{h^{n}_u}(\cdot)$ denotes the mathematical expectation in relation to $h^{n}_u$. $p^{n}_u$ characterizes the transmission power of device $u$, and $B_u^{\mathrm{UL}}$ is the allocated bandwidth for device $u$. $I^n_u$ signifies the interference, while $N_0$ represents the power spectral density of the noise. The channel gain $h^{n}_u$ is determined by
\begin{equation}\label{ChannelFading}
h^{n}_u = \varpi^{n}_u (d^{n}_u)^{-2},
\end{equation}
where  $d^{n}_u$ represents the distance between device $u$ and the edge server \cite{9210812, 10415249}, and $\varpi^{n}_u$ corresponds to the Rayleigh fading coefficient.

{In wireless communication, the presence of communication errors is an inherent reality that deviates from the hypothetical ideal conditions. Each device's local gradient is transmitted as a single packet, and a cyclic redundancy check (CRC) mechanism is employed to examine its content for potential errors. The packet error rate from device $u$ to the edge server is expressed as
\begin{equation}
q_u^{n}\left( p_u^{n}\right)= \mathbb{E}_{h_u^{n}}\left(1-\exp \left(-\frac{\Upsilon \left(I_u^{n}+B_u^{\mathrm{UL}} N_0\right)}{p_u^{n} h_u^{n}}\right)\right),
\end{equation}
where $\Upsilon$ is the waterfall threshold \cite{9658206}.
Let $\alpha_u^{n}$ indicate whether the local gradient $\boldsymbol{g}^{n}_u$ from device $u$ is received by the edge server at the $n$-th iteration, which is given by
\begin{equation}
\label{TransmittingSuccessfulyPro}
\alpha_u^{n}= \begin{cases}1, & \text { w.p. } 1-q_u^{n}\left(p_u^{n}\right), \\ 0, &\text { w.p. } q_u^{n}\left( p_u^{n}\right), \end{cases}
\end{equation}
where “w.p.” is the abbreviation of “with probability”.}

\subsection{Wireless Federated Learning Model}
The loss function to evaluate the training error  is given by
\begin{equation}
\label{LossFuction_Local}
F_{u}\left(\boldsymbol{w}_{u}\right)=\frac{1}{N_{u}} \sum_{i=1}^{N_{u}} f\left(\boldsymbol{w}_{u} ; \boldsymbol{x}_{u, i}, \boldsymbol{y}_{u, i}\right), \quad \forall u \in \mathcal{U},
\end{equation}
where  $\boldsymbol{w}_{u}$ denotes the local ML model's parameters of device $u$, and $f\left(\boldsymbol{w}_{u}; \boldsymbol{x}_{u, i}, \boldsymbol{y}_{u, i}\right)$ is the defined loss function determined by a specific ML task.

{
Moreover, let $N=\sum_{u=1}^{U}N_{u}$ and $\boldsymbol{w}$ represent the parameters of the global ML model. Then the loss function of the global ML model can be given by
\begin{equation}
\label{GlobalLossFuction}
\begin{aligned}
F\left(\boldsymbol{w}\right) & \triangleq  \frac{\sum_{u=1}^{U}{N_{u} F_{u}\left(\boldsymbol{w}_u\right)}}{\sum_{u=1}^{U}N_{u}} \\
&=\frac{\sum_{u=1}^{U}\sum_{i=1}^{N_{u}} f\left(\boldsymbol{w}_u ; \boldsymbol{x}_{u, i}, \boldsymbol{y}_{u, i}\right)}{\sum_{u=1}^{U}N_{u}}.
\end{aligned}
\end{equation}}
Actually, the whole process of FL is to address the following optimization problem
\begin{equation}
\label{MinTheGlobalLossFunction}
\boldsymbol{P}{1}: \boldsymbol{w}^{*}=\arg \min F\left(\boldsymbol{w}\right).
\end{equation}

Aim for  solving $\boldsymbol{P}{1}$,  each device performs GD algorithm in parallel. Since the local model is synchronized with the global model at the $n$-th round, it holds that $\boldsymbol{w}^n_u = \boldsymbol{w}^n, \forall u \in \mathcal{U}$. Each device then computes its local gradient by
\begin{equation}
\label{LocalGradientUpdating}
\boldsymbol{g}_u^n (\boldsymbol{w}^{n})=\frac{1}{N_u} \sum_{i=1}^{N_u} \nabla f\left(\boldsymbol{w}^{n} ; \boldsymbol{x}_{u, i}, \boldsymbol{y}_{u, i}\right), \quad u=1, \ldots, U.
\end{equation}
After that, each device transmits the newly obtained gradients to the edge server. Then, the edge server aggregates the received gradients as
 \begin{equation}
 \label{GlobalGradientAggregation}
\boldsymbol{g}^{n}(\boldsymbol{w}^{n}; \{p_u^{n}\})=\frac{\sum\limits_{u=1}^U N_u \alpha_u^{n} \boldsymbol{g}_u^{n} (\boldsymbol{w}^{n})}{\sum\limits_{u=1}^U N_u \alpha_u^{n}},
\end{equation}
and performs the global ML model update by
\begin{equation}
\label{LocalModelUpdating}
\boldsymbol{w}^{n+1}=\boldsymbol{w}^n-\eta \boldsymbol{g}^{n}(\boldsymbol{w}^{n}; \{p_u^{n}\}),
\end{equation}
where $\eta$ denotes the given learning rate. By solving Eq. (\ref{LocalGradientUpdating}) - Eq. (\ref{LocalModelUpdating}) iteratively,  the optimal FL model can be obtained finally.

\subsection{Model Pruning}
{For the purpose of avoiding weak devices from blocking the learning process, model pruning is introduced. Model pruning involves the removal of insignificant connections or filters to compress the model, thereby reducing computations with minimal loss in accuracy. According to \cite{9598845}\cite{molchanov2019importance}, at the $n$-th iteration, we can evaluate the importance of the $v$-th parameter  by
\begin{equation}
\label{ImportanceEvaluation}
I^n_{u,v}=\left(F^n_{u}\left(\mathcal{D}_{u}, \boldsymbol{w}^n_{u}\right)-F^n_{u}\left(\mathcal{D}_{u},\left.\boldsymbol{w}^n_{u}\right|_{w^n_{u,v}=0}\right)\right)^2,
\end{equation}
which is essentially the mean squared error with and without $w^n_{u,v}$. The smaller the value of $I^n_{u,v}$ is, the less important the  $w^n_{u,v}$.
Nevertheless, the calculation for obtaining the mean squared error $I^n_{u,v}$ is computation-intensive, especially when the model is large.

Therefore,  to alleviate the computational complexity, the approximate to Eq. (\ref{ImportanceEvaluation}) is introduced, which is denoted by
\begin{equation}
I^n_{u,v}=\left\| w^n_{u,v}\right\|.
\end{equation}
}

After the importance evaluation,  model pruning is performed by zeroing out the parameters having relatively low importance values. We define the pruning ratio $\rho^{n}_u$ of  device $u$ as
\begin{equation}
\rho^{n}_u=\frac{V^{n}_{u}}{V},
\end{equation}
where $V^{n}_u$ represents the number of the pruned parameters of the deployed model on device $u$. Let $\boldsymbol{\hat{w}}$ denote the pruned model, and ${\boldsymbol{g}}^{n}_u (\boldsymbol{\hat{w}}_u^{n})$ represent the local gradient computed based on the pruned model $\boldsymbol{\hat{w}}^n_{u}$. Hence, the processes to aggregate the local gradients and updating the global model, as shown in Eq. (\ref{GlobalGradientAggregation}) and Eq. (\ref{LocalModelUpdating}), can be reformulated as follows:
 \begin{equation}
\label{GlobalGradientAggregationwithPruning}
\boldsymbol{{g}}^{n}(\{\hat{\boldsymbol{w}}_u^{n}\}; \{\rho^{n}_u\}, \{p_u^{n}\})=\frac{\sum\limits_{u=1}^U N_u \alpha_u^{n} \boldsymbol{{g}}_u^{n} (\boldsymbol{\hat{w}}_u^{n})}{\sum\limits_{u=1}^U N_u \alpha_u^{n}},
\end{equation}
and
\begin{equation}
\label{LocalModelUpdatingwithPruning}
\boldsymbol{w}^{n+1}=\boldsymbol{w}^n-\eta \boldsymbol{{g}}^{n}(\{\hat{\boldsymbol{w}}_u^{n}\}; \{\rho^{n}_u\}, \{p_u^{n}\}),
\end{equation}
respectively. In addition to the local gradients, it is essential to transmit the pruning indices to the edge server for accurate aggregation.

\subsection{Gradient Quantization}
{
Considering the constraints on communication resources, reducing the amount of transmitted data is beneficial. To achieve this, we employ stochastic quantization \cite{9611373} for compressing the local gradients before sending them to the edge server.

Consider \( g^n_{u,v} \) as the \( v \)-th component of \( \boldsymbol{{g}}_u^{n}(\hat{\boldsymbol{w}}_u^n),  v \in \left\{1,2,\dots, V\right\} \), constrained by \( \left|g_{u,v}^{n}\right| \in \left[\underline{g}_{u,v}^n, \bar{g}_{u,v}^n\right] \). Here, \( \underline{g}_{u,v}^n \) and \( \bar{g}_{u,v}^n \) represent the minimum and maximum values of the \( v \)-th component of the gradient, respectively. Let \( \mathcal{Q} \) denote the quantization function and \( \delta_{u}^{n} \) the number of quantization bits for device \( u \). To perform quantization, the interval \( \left[\underline{g}_{u,v}^n, \bar{g}_{u,v}^n\right] \) is uniformly divided into \( 2^{\delta_{u}^{n}} \) segments, defined by the set \( \left\{b_0, b_1, \cdots, b_{2^{\delta_{u}^{n}}-1}\right\} \), where \( b_t \) is calculated as

\begin{equation}
b_t=\underline{g}_{u,v}^n+t \cdot \frac{\bar{g}_{u,v}^n-\underline{g}_{u,v}^n}{2^{\delta_{u}^{n}}-1}, \quad t=0, \cdots, 2^{\delta_{u}^{n}}-1.
\end{equation}
Furthermore, $g_{u,v}^{n}$  falling within the interval   $\left[b_t, b_{t+1}\right)$ can be quantized as
\begin{equation}
\mathcal{Q}\left(g_{u,v}^{n}\right)=\left\{\begin{array}{l}
\operatorname{sign}\left(g_{u,v}^{n}\right)\cdot b_t, \quad ~ \text { w.p. } \frac{b_{t+1}-\left|g_{u,v}^{n}\right|}{b_{t+1}-b_t}, \\
\operatorname{sign}\left(g_{u,v}^{n}\right) \cdot b_{t+1}, \quad  \text {w.p.} \frac{\left|g_{u,v}^{n}\right|-b_t}{b_{t+1}-b_t},
\end{array}\right.
\end{equation}
where \(\operatorname{sign}\left(\cdot\right)\) denotes the sign function, and “w.p.” is the abbreviation of “with probability”. Let \(\xi\) indicate the total number of bits required to encode \(\underline{g}_{u,v}^n\), \(\bar{g}_{u,v}^n\), and \(\operatorname{sign}\left(g_{u,v}^{n}\right)\). Therefore, the total bit count for the quantized local gradient is determined as
\begin{equation}\label{QuantizationRelationship}
\tilde{\delta_{u}^{n}} = V\delta_{u}^{n}+\xi.
\end{equation}

 Therefore,  Eq. (\ref{GlobalGradientAggregationwithPruning})  and Eq. (\ref{LocalModelUpdatingwithPruning})  can be further rewritten as
 \begin{equation}
 \label{GlobalGradientAggregationReality}
\boldsymbol{\overline{g}}^{n}(\{\hat{\boldsymbol{w}}_u^{n}\}; \{\rho_u^{n}\},\{\delta_u^{n}\},\{p_u^{n}\})=\frac{\sum\limits_{u=1}^U N_u \alpha_u^{n} \mathcal{Q}\left({\boldsymbol{g}}^{n}_u(\boldsymbol{\hat{w}}^n_u)\right)}{\sum\limits_{u=1}^U N_u \alpha_u^{n} },
\end{equation}
and
\begin{equation}
\label{ModelUpdateReality}
\boldsymbol{{w}}^{n+1}=\boldsymbol{w}^n-\eta \boldsymbol{\overline{g}}^{n}(\{\hat{\boldsymbol{w}}_u^{n}\}; \{\rho_u^{n}\},\{\delta_u^{n}\},\{p_u^{n}\}),
\end{equation}
respectively.}
\begin{remark}
The gradient aggregation in Eq. (\ref{GlobalGradientAggregationReality}) exhibits a loss of accuracy due to pruning error, quantization error, and transmission error, which contrasts with the ideal gradient aggregation. As a result, this degrades the FL's convergence performance.
\end{remark}

\section{Convergence Analysis of FL}
In this section, we present a comprehensive analysis to assess the influence of model pruning error, transmission error, and quantization error on the FL's convergence.
\subsection{Basic Assumptions}
{Prior to delving into the convergence analysis, we first introduce widely accepted assumptions, as outlined below:
\begin{itemize}
  \item   \begin{assumption}\label{CompleteAssumption1}
        $\nabla F(\boldsymbol{w})$ is uniformly $L$-Lipschitz continuous in reference to $\boldsymbol{w}$, which is represented as
\begin{equation}
\left \|\nabla F\left(\boldsymbol{w}^{n+1}\right)-\nabla F\left(\boldsymbol{w}^{n}\right)\right \| \leq L\left \|\boldsymbol{w}^{n+1}-\boldsymbol{w}^{n}\right \|,
\end{equation}
where $L$ is the Lipschitz constant associated with $F(\cdot)$.
  \end{assumption}
\item
\begin{assumption}\label{CompleteAssumption3}
 $\nabla F(\boldsymbol{w})$ is twice-continuously differentiable. Given Assumption \ref{CompleteAssumption1} , we can obtain
\begin{equation}
\label{Assumtion3}
\nabla^{2} F(\boldsymbol{w}) \preceq L \boldsymbol{I},
\end{equation}
where $\boldsymbol{I}$ denotes an identity matrix.
  \end{assumption}
 \item

   \begin{assumption}\label{CompleteAssumption4}
  The second moments of parameters are constrained by
 \begin{equation}
\mathbb{E}\left\{\|\boldsymbol{w}\|^2\right\} \leq D^2.
\end{equation}
  \end{assumption}
\item
  \begin{assumption}\label{CompleteAssumption5}
  It is assumed that
\begin{equation}\label{EqAssumption5}
\left\|\nabla f\left(\boldsymbol{w}^{n}; \boldsymbol{x}_{u,i}, \boldsymbol{y}_{u,i}\right)\right\|^2 \leq {\upsilon_{1}+\upsilon_{2}\left\|\nabla F\left(\boldsymbol{w}^{n}\right)\right\|^{2}},
\end{equation}
where $\upsilon_{1},\upsilon_{2}\geq0$.

    \end{assumption}
\end{itemize}

It is worth noting that many previous works adopt the aforementioned assumptions for theoretical analysis \cite{9912341},\cite{9292468},\cite{10229183}. Nevertheless, in fact, these strategies remain effective even when the assumptions do not hold completely. }

\subsection{Convergence Analysis}
{First, we present the following Lemmas \ref{QuantizationLemma} to \ref{ModelPruningLemma} to characterize the impact of quantization and model pruning.}
{
\begin{lemma}
\label{QuantizationLemma}
Through stochastic quantization, each quantized gradient satisfies\cite{quantize}\cite{amiri2020federated}
\begin{equation}\label{UnbiasedQuantization}
\mathbb{E}\left[\mathcal{Q}\left(\boldsymbol{g}_u^{n}(\hat{\boldsymbol{w}}_u^n)\right)\right]=\boldsymbol{g}_u^{n}(\hat{\boldsymbol{w}}_u^n),
\end{equation}
 and consequently the quantization error satisfies 
\begin{equation}
\label{QuantizationErrorBound}
\mathbb{E}\left[\left\|\mathcal{Q}\left(\boldsymbol{g}_u^{n}(\hat{\boldsymbol{w}}_u^n)\right)-\boldsymbol{g}_u^{n}(\hat{\boldsymbol{w}}_u^n)\right\|^2\right] \leq \frac{\sum\limits_{v=1}^V\left(\bar{g}_{u,v}^n-\underline{g}_{u,v}^n\right)^2}{4\left(2^{\delta_u^n}-1\right)^2}.
\end{equation}
\end{lemma}
\begin{proof}
See the detail proof in \cite{9277666}.
\end{proof}}

\begin{lemma}
\label{ModelPruningLemma}
The model pruning error with pruning ratio $\rho_u^n $ is represented as
\begin{equation}
\label{ModelPruningErrorBound}
\mathbb{E}\left\{\left\|\boldsymbol{w}_u^n-\hat{\boldsymbol{w}}_u^n\right\|^2\right\} \leq \rho_u^n \mathbb{E}\left\{\left\|\boldsymbol{w}_u^{n}\right\|^2\right\} \leq \rho_u^n D^2.
\end{equation}
\end{lemma}
\begin{proof}
See the detail proof  in  \cite{stich2018sparsified}.
\end{proof}

Based on Assumptions \ref{CompleteAssumption1} to \ref{CompleteAssumption5} and  Lemmas \ref{QuantizationLemma} and \ref{ModelPruningLemma},  we can obtain the convergence rate of FL by Theorem \ref{Theorem1}.
{\begin{theorem}
\label{Theorem1}
When the model pruning strategy $\{\rho_u^n\}$,  quantization level strategy $\{\delta_u^{n}\}$,  power control strategy $\{p_u^{n}\}$ are given, the upper bound on the average $\ell_2$-norm of the gradients after $\Omega$ iterations with the learning rate $\eta =\frac{1}{L}$ is represented as
\begin{equation}\label{TheoremReConvergenceRate}
\begin{aligned}
 & \frac{1}{\Omega+1}\sum_{n=0}^{\Omega}\mathbb{E}\left\{\left\|\nabla F\left({\boldsymbol{w}}^n\right)\right\|^2\right\} \\ & \leq  \frac{2L}{(1-12\upsilon_2)(\Omega+1)}\mathbb{E}\left\{F\left(\boldsymbol{w}^0\right)-F\left(\boldsymbol{w}^{*}\right)\right\}\\& + \frac{1}{(\Omega+1)} \sum_{n=0}^{\Omega} \Gamma^n,
\end{aligned}
\end{equation}
where $\Gamma^n$ is given by
\begin{equation}
\begin{aligned}
\Gamma^n = & \frac {1}{1-12\upsilon_2} \left(3 \sum _ { u = 1 }^{U}\! \frac{\sum_{v=1}^V\left(\bar{g}_{u, v}^n-\underline{g}_{u, v}^n\right)^2}{4\left(2^{\delta_{u}^n}-1\right)^2} \right. \\ &+ \left. 3L^2 D^2 \sum_{u=1}^{U} \rho_u^n + \frac{12\upsilon_1}{N} \sum_{u=1}^U N_u q_u^n \right).
\end{aligned}
\end{equation}
\end{theorem}}
{
\begin{proof}
See  Appendix \ref{ProofTheorem}.
\end{proof}}
{
\begin{remark}
The $\ell_2$-norm of the average gradients, as presented in Theorem \ref{Theorem1}, is commonly used to indicate convergence rate. Specifically, if $\frac{1}{\Omega+1}\mathbb{E}\left\{\left\|\nabla F\left({\boldsymbol{w}}^n\right)\right\|^2\right\} \leq \zeta$, it implies that FL has achieved a $\zeta$-optimal solution.
\end{remark}}
{
As ${\Omega \rightarrow \infty}$, the result of Eq. (\ref{TheoremReConvergenceRate}) converges to
\begin{equation}\label{ErrorGap}
\begin{aligned}
&  \frac{1}{\Omega+1}\sum_{n=0}^{\Omega}\mathbb{E}\left\{\left\|\nabla F\left({\boldsymbol{w}}^n\right)\right\|^2\right\} \leq  \frac{1}{\Omega+1}\sum_{n=0}^{\Omega}  \Gamma^n.
\end{aligned}
\end{equation}
We can observe a convergence gap attributed to errors in model pruning, gradient quantization, and wireless transmission.}
{
\begin{remark}
To mitigate the deterioration of convergence, it is imperative to minimize the convergence gap as depicted in Eq. (\ref{ErrorGap}). It is obvious that minimizing the convergence gap can be achieved by adjusting model pruning strategy, gradient quantization strategy, and transmission power control strategy to minimize $\Gamma^n$.
\end{remark}}
\section{Performance Evaluation And Problem Formulation}
As previously highlighted, the methodologies of model pruning, gradient quantization, and transmission power control play a critical role in FL's convergence process. Besides, the aforementioned methodologies also exert a substantial impact on energy consumption and time overhead during the pursuit of convergence in FL. Hence, it becomes imperative to meticulously calibrate these strategies to ensure an equilibrium between convergence efficiency and associated training overheads, namely, energy consumption and delay.
\subsection{Delay Model}\label{DelayModelSubsection}
The delay predominantly stems from three distinct components. Specifically,
\subsubsection{Delay in Local Training}
The delay for training local model on device $u$  during the $n$-th iteration can be represented as
\begin{equation}
\label{LocalComputingLatency}
T_{u,lt}^{n}=\frac{{N}_{u} c_{0}\left(1-\rho_u^n\right)}{f_{u}^{n}},
\end{equation}
where $c_{0}$ (cycles/sample) denotes the requisite CPU cycles for training a single data sample  via backpropagation algorithm on device $u$, and $f^{n}_u $  signifies the available computational resources of device $u$. 
\subsubsection{Delay in Local Gradient Uploading}
The delay for uploading the updated local gradient to edge server is given by\footnote{Given the relatively smaller data size, the delay incurred by uploading pruning indexes to the edge server can be deemed negligible \cite{9598845}.}
\begin{equation}
\label{LocalUploadingLatency}
T_{u,lu}^{n} = \frac{\tilde{\delta_{u}^{n}}\left(1-\rho_u^n\right)}{R_u^{n}\left(p^{n}_u\right)}.
\end{equation}
\subsubsection{Delay in Gradient Aggregation,  Model Updating and Broadcast}
Upon receipt of all local gradients from devices, the edge server undertakes gradient aggregation and model updating. Subsequently, it broadcasts the refreshed global model to each device. The cumulative delay encompassing the aforementioned operations on the edge server is represented as
\begin{equation}
\label{GlobalBroadcastLatency}
T^{n}_{gb}=s,
\end{equation}
where $s$ is a small constant, owing to the powerful computing capability and high transmission power of the edge server.

Consequently, the aggregate time overhead for a single iteration can be delineated as
\begin{equation}
T^{n}=\max _{u \in \mathcal{U}}\left\{T_{u,lt}^{n}+T_{u,lu}^{n} \right\}+T^{n}_{gb} .
\end{equation}
\subsection{Energy Consumption Model}\label{EnergyConsumptionModelSubsection}
In our analysis, we confine our focus to the energy consumption associated with the devices, excluding the edge server's consumption given its consistent and uninterrupted energy source. The predominant factors contributing to this device-centric energy expenditure are local training and the process of uploading local gradients. The energy consumed for local training of device $u$ can be expressed as
\begin{equation}
\label{LocalComputingEngeryConsumption}
E_{u,lt}^{n} = k(f_{u}^{n})^{\sigma}T_{u,lt}^n,
\end{equation}
where $k $ and $\sigma$ are constant parameters \cite{10051719}. Concurrently, the energy consumption attributed to the local gradient uploading can be calculated as
\begin{equation}
\label{LocalCommunicationEngeryConsumption}
E_{u, l u}^{n}= p_u^{n}T_{u, l u}^n.
\end{equation}

Thus, the energy consumption of device $u$ for a single iteration can be given by
\begin{equation}
\label{LocalEngeryConsumption}
E_{u}^{n}= E_{u,lt}^{n}+E_{u, l u}^{n} .
\end{equation}
\subsection{Problem Formulation}
The optimization problem is conceived as
\begin{subequations}\label{OptimizationProblem}
\begin{equation}
\label{OptimizationFunctionP1}
\begin{aligned}
\mathcal{P}1: \min _{\boldsymbol{\delta}^{n}, \boldsymbol{p}^{n},\boldsymbol{\rho}^n }\Gamma^n
\end{aligned}
\end{equation}
\begin{equation}
\label{Tmax}
{s.t.}\quad \quad T^n \leq T^{\rm max}, \forall n,
\end{equation}
\begin{equation}
\label{Emax}
 \quad \quad \quad E_u^n \leq E^{\rm max},\forall u, \forall n,
\end{equation}
\begin{equation}
\label{QuantizationLevel}
 \quad \quad \quad \delta_u^{n} \in \mathbb{Z}^{+} , \forall u, \forall n,
\end{equation}
\begin{equation}
\label{QuantizationLevel_upperbound}
 \quad \quad \quad \delta_u^{n} \leq \delta^{\rm max}, \forall u, \forall n,
\end{equation}
\begin{equation}
\label{TransmittingPowerConfiguration}
  \quad \quad \quad p^{\rm min} \leq p_u^n \leq p^{\rm max},\forall u, \forall n,
\end{equation}
\begin{equation}
\label{ModelPruningRatioCons}
  \quad \quad \quad 0 \le \rho_u^n \leq \rho^{\rm max},\forall u, \forall n,
\end{equation}
\end{subequations}
where  $\boldsymbol{\delta}^{n}=\left[\delta^{n}_1, \cdots, \delta^{n}_U\right]$, $\boldsymbol{p}^{n}=\left[p^{n}_1, \cdots, p^{n}_U\right]$, and $\boldsymbol{\rho}^{n}=\left[\rho^{n}_1, \cdots, \rho^{n}_U\right]$   represents quantization strategy,  power control strategy, and model pruning strategy, respectively.  Eq. (\ref{Tmax}) and Eq. (\ref{Emax}) are the time and energy consumption constraints for each iteration, respectively.  $\mathbb{Z}^{+}$ in Eq. (\ref{QuantizationLevel}) is the positive integer set. $p^{\rm min}$ and $p^{\rm max}$ represent the lower and upper limits of transmission power, respectively. $ \delta^{\rm max}$ and $\rho^{\rm max}$  are the upper limits of the quantization level and model pruning ratio, respectively.

\section{Algorithm Design}
Evidently,  $\mathcal{P}1$ exhibits non-convexity due to the presence of coupled integer and continuous variables. Therefore, we decompose it into several subproblems.
\subsection{Optimal Model Pruning Strategy}\label{SectionModelPruning}
When the quantization strategy $\boldsymbol{\delta}^{n}$ and power control strategy $\boldsymbol{p}^{n}$ are given, problem $\mathcal{P}1 $ becomes
\begin{subequations}\label{SubProblem1}
\begin{equation}
\begin{aligned}
\mathcal{P}2: \min _{\boldsymbol{\rho}^n } ~~ \Gamma^n (\boldsymbol{\rho}^n )
\end{aligned}
\end{equation}
\begin{equation}
\label{ConsSubP1}
{s.t.}~~\textrm{(\ref{Tmax}), (\ref{Emax}),  (\ref{ModelPruningRatioCons})}.
\end{equation}
\end{subequations}
$\mathcal{P}2 $ is a linear programming problem with respect to $\rho_u^{n}, \forall u, \forall n$.  Therefore, we can obtain the optimal model pruning strategy by Theorem \ref{theorem2}.
\begin{theorem}\label{theorem2}
 Given quantization strategy $\boldsymbol{\delta}^{n}$, transmission power control strategy $\boldsymbol{p}^{n}$, the optimal model pruning strategy $(\rho_u^{n})^*, \forall u, \forall n $ can be given by
\begin{equation}\label{rho}
(\rho_u^{n})^*= \min\left\{\rho^{\rm max}, \left(1- \min\left\{ \Phi_1,\Phi_2     \right\} \right)^{+}\right\},
\end{equation}
where $(x)^{+}=\max\left\{0,x\right\}$, and  $\Phi_1$ and $\Phi_2$ are represented as
\begin{equation}\label{phi1}
\Phi_1= \frac{T^{\rm max}-s}{\frac{N_uc_0}{f_u^n}+\frac{\tilde{\delta}^n_u}{R_u^{\mathrm{n}}\left( p^{n}_u\right)}} ,
\end{equation}
and
\begin{equation}\label{phi2}
\Phi_2= \frac{E^{\rm max}}{k(f_u^n)^{\sigma-1}N_uc_0+\frac{p_u^n\tilde{\delta}_u^n}{R_u^{\mathrm{n}}\left( p^{n}_u\right)}},
\end{equation}
respectively.
\end{theorem}
\begin{proof}
See Appendix \ref{ProofTheorem2}.
\end{proof}
\subsection{Optimal Quantization Strategy}\label{SectionGradientQuantization}
When given power control strategy $\boldsymbol{p}^{n}$ and the optimal model pruning strategy $\boldsymbol{\rho}^{n}$ obtained by Theorem \ref{theorem2},  the optimization problem can be reformulated as
\begin{subequations}\label{SubProblem2}
\begin{equation}
\begin{aligned}
\mathcal{P}3: \min _{\boldsymbol{\delta}^n } ~~ ~~ \Gamma^n (\boldsymbol{\delta}^n ; \{\rho_u^n\}^*)
\end{aligned}
\end{equation}
$${s.t.}~~~\textrm{(\ref{QuantizationLevel}), (\ref{QuantizationLevel_upperbound})},$$
\begin{equation}
\begin{aligned}
\label{TmaxSub2}
T^{n}\left(\boldsymbol{\delta}^n ; \{\rho_u^n\}^* \right)\leq T^{\rm max}  ,\forall u, \forall n,
\end{aligned}
\end{equation}
\begin{equation}
\label{EmaxSub2}
 E_u^{n}\left(\boldsymbol{\delta}^n ; \{\rho_u^n\}^* \right)\leq E^{\rm max}  ,\forall u, \forall n,
\end{equation}
\end{subequations}
where $\Gamma^n (\boldsymbol{\delta}^n ;\{\rho_u^n\}^*)$, $T^{n}\left(\boldsymbol{\delta}^n ; \{\rho_u^n\}^* \right)$, and $E_u^{n}\left( \boldsymbol{\delta}^n ;\{\rho_u^n\}^* \right)$ are obtained by applying the optimal  model pruning strategy $\{\rho_u^n\}^* $ into  Eq. (\ref{OptimizationFunctionP1}),  Eq. (\ref{Tmax}), and Eq. (\ref{Emax}),
respectively. For further discussion, we give following lemma to discuss the monotonicity of problem $\mathcal{P}3$.
\begin{lemma}\label{Lemma3}
 Problem $\mathcal{P}3$ is monotonically decreasing with respect to variable  ${\boldsymbol{\delta}^n }$ in the domain.
\end{lemma}
\begin{proof}
See Appendix \ref{Lemma3Proof}.
\end{proof}

Consequently, we can obtain the optimal solution of quantization strategy $(\delta_u^n)^*, \forall u, \forall n$ from Theorem \ref{theorem3}.

\begin{theorem}\label{theorem3}
 Given the optimal model pruning strategy $(\rho_u^{n})^*$, and  transmission power configuration strategy $\boldsymbol{p}^{n}$,  we can obtain the optimal quantization strategy $(\delta_u^n)^*, \forall u, \forall n$ as
\begin{equation}
\label{OptimalQuantizationStrategy}
(\delta_u^n)^* =\left\lceil \min\left\{ \frac{\Phi_3-\xi}{V},\frac{\Phi_4-\xi}{V} , \delta^{\rm max} \right\}\right\rceil,
\end{equation}
where  $\lceil x \rceil$  represents the minimum positive integer that is less than or equal to $x$.   $\Phi_3$ and $\Phi_4$ are represented as
\begin{equation}
\begin{aligned}
\label{TmaxSub2ExpanisionRePhiSystemModel}
\Phi_3 = \frac{\left(T^{\rm max}-s-\frac{{N}_{u} c_0\left(1-  (\rho_u^n)^*\right)}{f_u^n}\right)R_u^{\mathrm{n}}\left( p^{n}_u\right) }{1-  (\rho_u^n)^*},
\end{aligned}
\end{equation}
and
\begin{equation}
\label{EmaxSub2ExpanisionRePhiSystemModel}
\Phi_4 = \frac{ \left(E^{\rm max}-k(f_u^n)^{\sigma-1}{N}_{u}c_0\left(1-(\rho_u^n)^*\right)\right)R_u^{\mathrm{n}}\left( p^{n}_u\right)}{p_u^n\left(1-(\rho_u^n)^*\right)},
\end{equation}
respectively.
\end{theorem}
\begin{proof}
See Appendix \ref{Theorem3Proof}.
\end{proof}

\subsection{Optimal Transmission Power Configuration Strategy}\label{SectionPowerControl}
When given optimal model pruning strategy $\boldsymbol{\rho}^{n}$ and  gradient quantization strategy $\boldsymbol{\delta}^{n}$, the problem becomes
\begin{subequations}\label{SubProblem3}
\begin{equation}
\begin{aligned}
\mathcal{P}4: \min _{\boldsymbol{p}^{n}} ~~   \Gamma^n (\boldsymbol{p}^{n}; \{\rho_u^n\}^*, \{\delta_u^n\}^*)
\end{aligned}
\end{equation}
$${s.t.}~~~\textrm{(\ref{TransmittingPowerConfiguration})},$$
\begin{equation}
\begin{aligned}
\label{TmaxSub3}
T^{n}\left(\boldsymbol{p}^{n}; \{\rho_u^n\}^*, \{\delta_u^n\}^* \right)\leq T^{\rm max}  ,\forall u, \forall n,
\end{aligned}
\end{equation}
\begin{equation}
\label{EmaxSub3}
 E_u^{n}\left(\boldsymbol{p}^{n}; \{\rho_u^n\}^*, \{\delta_u^n\}^* \right)\leq E^{\rm max}  ,\forall u, \forall n,
\end{equation}
\end{subequations}
where $\Gamma^n (\boldsymbol{p}^{n};\{\rho_u^n\}^*,  \{\delta_u^n\}^*)$, $T^{n}\left(\boldsymbol{p}^{n}; \{\rho_u^n\}^* , \{\delta_u^n\}^*\right)$, and $E_u^{n}\left( \boldsymbol{p}^{n}; \{\rho_u^n\}^* , \{\delta_u^n\}^*\right)$ are obtained by applying the optimal  model pruning strategy $\{\rho_u^n\}^* $ and the optimal gradient quantization strategy  $\{\delta_u^n\}^* $ into  Eq. (\ref{OptimizationFunctionP1}),  Eq. (\ref{Tmax}), and Eq. (\ref{Emax}), respectively.
Before solving problem $\mathcal{P}4$, we investigate the convexity of  $\mathcal{P}4$ in Lemma \ref{Lemma4}.
\begin{lemma}\label{Lemma4}
Problem $\mathcal{P}4$ is not always convex, which is determined by the optimal model pruning strategy $\{\rho_u^n\}^*, \forall u $.
\end{lemma}
\begin{proof}
This conclusion can be readily deduced through the analysis of the properties of its Hessian matrix \cite{boyd2004convex}.
\end{proof}


To solve $\mathcal{P}4$, we design a Bayesian optimization-based algorithm \cite{7352306}  to obtain a near-optimal solution with low complexity. For convenience, in the following, we use $\Gamma^n(\boldsymbol{p}^n)$ to represent $\Gamma^n (\boldsymbol{p}^{n};\{\rho_u^n\}^*,  \{\delta_u^n\}^*)$.

The fundamental concept of Bayesian optimization lies in using a surrogate model to approximate the objective function. By employing an acquisition function,  new points are sampled and evaluated in a cost-effective approach. Consequently, the surrogate model updates its prior knowledge with these newly obtained samples and iteratively explores the optimal solution of the objective function.

\subsubsection{Surrogate Model}
We use Gaussian Process (GP) to establish the surrogate model $\hat \Gamma^n(\boldsymbol{p}^n)$ as $\hat\Gamma^n \sim \mathcal{GP}(0,\kappa(\boldsymbol{p}^n,(\boldsymbol{p}^n)'))$, where $\kappa(\boldsymbol{p}^n,(\boldsymbol{p}^n)')$ is a kernel (covariance) function measuring pairwise similarity of  $\boldsymbol{p}^n$ and $(\boldsymbol{p}^n)'$.
Then, the joint prior probability density funciton (PDF) of function evaluations $\boldsymbol{\hat\Gamma}^n_M=[\hat\Gamma^n(\boldsymbol{p}^n_1),...,\hat\Gamma^n(\boldsymbol{p}^n_M)]^\top $ at $\boldsymbol{P}^n_M=\{\boldsymbol{p}^n_1,...,\boldsymbol{p}^n_M\} $ follows a jointly Gaussian distribution as
\begin{equation}\label{GPPRIOR}
    p(\boldsymbol{\hat\Gamma}^n_M) = \mathcal{N}(\boldsymbol{\hat\Gamma}^n_M;0,\boldsymbol{K}_M),
\end{equation}
where $\boldsymbol{K}_M$ is an $M\times M$ covariance matrix with the $(i,j)$-th entry $[\boldsymbol{K}_M]_{i,j}=cov(\hat\Gamma^n_{i},\hat\Gamma^n_{j}):=\kappa (\boldsymbol{p}^n_i,\boldsymbol{p}^n_j)$.
For brevity, let $\mathcal{D}_M=\{(\boldsymbol{p}^n_{i},\iota^n_{i})\}_{i=1}^M$ represent all the sampled data at the $M$-th iteration, where $\iota^n_i=\Gamma^n(\boldsymbol{p}^n_i)$. Along with the GP prior in Eq. (\ref{GPPRIOR}), the posterior PDF $p(\hat\Gamma(\boldsymbol{p}^n)|\mathcal{D}_M)$ can be formulated via Bayes' rule as
\begin{equation}
    p(\hat\Gamma^n(\boldsymbol{p}^n)|\mathcal{D}_M) = \mathcal{N}(\hat\Gamma^n(\boldsymbol{p}^n);\mu_M(\boldsymbol{p}^n),\sigma^2_M(\boldsymbol{p}^n)),
\end{equation}
where the mean and variance can be given by
\begin{equation}
  \mu_M(\boldsymbol{p}^n) = \boldsymbol{k}_M^\top (\boldsymbol{p}^n)\boldsymbol{K}_M^{-1}\boldsymbol{\iota}^n_M
\end{equation}
and
\begin{equation}
  \sigma_M^2(\boldsymbol{p}^n) = \kappa(\boldsymbol{p}^n, \boldsymbol{p}^n)-\boldsymbol{k}_M^\top (\boldsymbol{p}^n)\boldsymbol{K}_M^{-1}\boldsymbol{k}_M,
\end{equation}
where $\boldsymbol{k}_M(\boldsymbol{p}^n)=[\kappa(\boldsymbol{p}^n_1,\boldsymbol{p}^n),...,\kappa(\boldsymbol{p}^n_M,\boldsymbol{p}^n)]^\top $, and $\boldsymbol{\iota}^n_M = [\iota^n_1,...,\iota^n_M]^\top $. The kernel function $\kappa(\boldsymbol{p}^n,(\boldsymbol{p}^n)')$ can be calculated as
\begin{equation}
    \kappa(\boldsymbol{p}^n,(\boldsymbol{p}^n)')= exp(-\frac {||\boldsymbol{p}^n-(\boldsymbol{p}^n)'||^2}{2}).
\end{equation}
\subsubsection{Acquisition Function}
The acquisition function $\nu(x)$ is
\begin{equation}
\begin{aligned}
    \nu(\boldsymbol{p}^n)&=P(F(x)\leq \Gamma^n((\boldsymbol{p}^n)^*_M)+\varsigma )\\&=1-\Phi(\frac{\mu(\boldsymbol{p}^n)-\Gamma^n((\boldsymbol{p}^n)^*_M)-\varsigma }{\sigma(\boldsymbol{p}^n)}),
\end{aligned}
\end{equation}
where $P(\cdot)$ represents probability. $\varsigma $ is a positive hyper-parameter, and $(\boldsymbol{p}^n)^*_M$ indicates the point corresponding to the minimum function value in the currently obtained data $\mathcal{D}_M$, which can be denoted by
\begin{equation}
    (\boldsymbol{p}^n)^*_M=\mathop{argmin}\limits_{\boldsymbol{p}^n_i\in\boldsymbol{P}^n_M}\Gamma^n(\boldsymbol{p}^n_i).
\end{equation}
$\Phi(x)$ is the cumulative distribution function  of normal distribution, which is defined as
\begin{equation}
    \Phi(x)=\int_{-\infty}^{x}\frac{1}{\sqrt{2\pi}}exp(-\frac{x^2}{2}).
\end{equation}
Using the acquisition function $\nu(x)$, the next sampling point $\boldsymbol{p}^n_{M+1}$ can be determined by
\begin{equation}\label{sample}
    \boldsymbol{p}^n_{M+1}=\mathop{argmax}\limits_{\boldsymbol{p}^n}\nu(\boldsymbol{p}^n).
\end{equation}
Based on this, we choose the point that is most probable to be less than the current minimum value $\Gamma^n((\boldsymbol{p}^n)^*_M)$ as the next sampling point. After acquiring the subsequent sample, we proceed to update $\mathcal{D}_M$ to $\mathcal{D}_{M+1}$ and advance to the next iteration. The iteration process continues until an optimal solution is achieved or the maximum number of iterations is reached.

\subsection{Joint Optimization on Model Pruning, Gradient Quantization, and Power control}
{In previous sections, we have discussed the optimal solutions of model pruning strategy, gradient quantization strategy, and power control strategy, respectively.
In this section, we combine them to realize a two-stage joint scheduling algorithm. Specifically, at the $k$-th iteration, the optimal model pruning strategy $\left(\rho_u^n\right)^*_k, \forall u$,  can be obtained according to Theorem \ref{theorem2},
based on the fixed gradient quantization strategy $(\delta_u^{n})^*_{k-1}, \forall u$,  and power control strategy $(p_u^{n})^*_{k-1}, \forall u$.  Furthermore,  we  obtain an optimal gradient quantization strategy $\left(\delta_u^n\right)^*_k, \forall u$ by Theorem \ref{theorem3},
relying on the newly updated optimal model pruning strategy $\left(\rho_u^n\right)^*_k$  and the fixed power control strategy $(p_u^{n})^*_{k-1}, \forall u$. Then, relying on the updated optimal model pruning strategy $\left(\rho_u^n\right)^*_k$ and the optimal gradient quantization strategy $\left(\delta_u^n\right)^*_k, \forall u$,
we can obtain the optimal power control strategy $\left(p_u^n\right)^*_k, \forall u$ by solving problem $\mathcal{P}4$. In the next iteration, we obtain the new model pruning strategy $\left(\rho_u^{n}\right)^*_{k+1}, \forall u$ based on the  latest $\left(p_u^n\right)^*_k, \forall u$ as well as $\left(\delta_u^n\right)^*_k, \forall u$.
Similarly, the update of the  gradient quantization and power control strategies will utilize the other new updated values of the fixed variables. The aforementioned procedure will be operated iteratively, until the following convergence criterion is satisfied
\begin{equation}
  \label{gap}
  \begin{aligned}
    \Gamma^{n}& \left(\left(\rho_u^{n}\right)^*_{k+1},   \left(\delta_u^{n}\right)^*_{k+1}, \left(p_u^{n}\right)^*_{k+1}\right)- \\& \Gamma^n\left(\left(\rho_u^n\right)^*_{k},  \left(\delta_u^n\right)^*_{k}, \left(p_u^n\right)^*_{k}\right) \leq \varrho.
  \end{aligned}
\end{equation}
Algorithm \ref{alg1} summarizes the above procedures.

The computational complexity of the proposed algorithm is $O(U)$ with a fixed number of iterations, where $U$ represents the number of devices. As the complexity increases linearly with the dimensionality of the input variables, the computational overhead remains low. Moreover, since channel conditions and device resources typically remain stable during training, the proposed two-stage algorithm is generally executed only once during initialization. Given that the algorithm runs on edge servers equipped with ample computational resources, its execution is efficient. In scenarios where channel conditions exhibit substantial variations, the algorithm can be re-executed multiple times. The execution frequency—whether once, periodically, or per iteration—can be flexibly adjusted by monitoring the degree of variation in channel conditions and node resources. This adaptability underscores the practical feasibility of the proposed algorithm in real-world deployments.}
{
    \begin{algorithm}[t]
        \caption{Control Algorithm for LTFL}
        \label{alg1}
        {
        \begin{algorithmic}
            \REQUIRE $T^{\rm max},E^{\rm max},p^{\rm min},p^{\rm max}, \rho^{\rm max},\delta^{\rm max}, U, N_u$
             \ENSURE  $\left(\rho_u^{n}\right)^*_k$, $\left(\delta_u^{n}\right)^*_k$, $\left(p_u^{n}\right)^*_k$
            \WHILE{$\text{gap} \geq \varrho$}
            \STATE Obtain  $(\rho_u^{n})^*_k$ by Eq. (\ref{rho}) using $(p_u^{n})^*_{k-1}$.
            \STATE Obtain $(\delta_u^n)^*_k$ by Eq. (\ref{OptimalQuantizationStrategy}) using $(p_u^{n})^*_{k-1}$ and $(\rho_u^{n})^*_k$.
             \STATE Obtain a  randomized sample pair $\left(\boldsymbol{p}^n_1, \iota^n_1 = \Gamma^n(\boldsymbol{p}^n_1)\right)$.
            \FOR{$M = 1$ to $M^{\rm max}$}
            \STATE Obtain the minimum point $((\boldsymbol{p}^n)^*_{k,M},(\iota^n)^*_{k,M}) \in \mathcal{D}_M$.
            \STATE Update the surrogate model $p(\hat{\Gamma}(\boldsymbol{p}^n)|\mathcal{D}_M)$.
            \STATE Calculate the acquisition function $\nu(\boldsymbol{p}^n)$.
            \STATE Determine next sampling point $\boldsymbol{p}^n_{M+1}$ by Eq. (\ref{sample}).
            \STATE Calculate $\iota^n_{M+1} = \Gamma^n(\boldsymbol{p}^n_{M+1})$.
            \STATE $M = M + 1$
            \ENDFOR
            \STATE $(\boldsymbol{p}^n)^*_k = [(p^n_1)^*_k,...,(p^n_U)^*_k]=(\boldsymbol{p}^n)^*_{k,M}$
            \STATE Calculate the gap defined in Eq. (\ref{gap}).
            \STATE $k = k + 1$
            \ENDWHILE

        \end{algorithmic}
        }
    \end{algorithm}
}

\section{Experimental Results}
\subsection{Experiment Settings}
{We have executed a series of experiments utilizing the CIFAR 10 dataset \cite{krizhevsky2009learning}. This dataset comprises 60,000 images, divided into 50,000 labeled samples for training purposes and 10,000 samples designated for testing. These images are organized into 10 distinct categories, each image measuring 32 by 32 pixels. We utilized this dataset to train a neural network for the precise classification of image recognition. Our designed Residual neural network begins with an initial convolutional layer that uses 64 3x3 convolutional kernels to perform convolution operations on the input 32x32x3 color images. This is followed by four groups of residual blocks, each consisting of multiple pre-activation residual blocks, with each block's convolutional layers having different convolutional kernels. After passing through these residual block groups, the network performs global average pooling, reducing the feature map size to 1x1x512.  \footnote{The detailed design specifics of the ML model are beyond the scope of this paper. However, it is important to note that the proposed LTFL framework demonstrates exceptional adaptability across various ML models.} For our experiments, we have deployed 30 devices, each containing a differing quantity of data samples, distributed uniformly in the range of 400 to 600. The main parameters used in our experiments are delineated in TABLE \ref{ParameterSetting}.}


\begin{table}[]
\caption{Parameter Settings}
\label{ParameterSetting}
\centering
{
\begin{tabular}{cccc}
\toprule
\textbf{Parameter}    & \textbf{Value}                &\textbf{ Parameter  }          & \textbf{Value   }                                     \\ \midrule
$p^{\rm max}$    & 0.1W                   & $I_u^{n}$            & $\mathrm{U}[10^{-8}, 2\times10^{-8}]$                 \\
$B_u^{\rm UL}$ & 10MHz                    & $\varpi_u^n$         & $0.015$                             \\
$\delta^{\rm max}$    & 8                    & $f_u^{n}$            & $\mathrm{U}[30, 110]\textrm{MHz}$                   \\
$p^{\rm min}$    & 0.01W                  & $\rho^{\rm max}$                   &0.5           \\
$N_0$        & -174 dbm/Hz             & $k$                  & $1.25 \times 10 ^{-26}$                      \\
$\Upsilon$          & 0.023dB                 & $\sigma$             & 3                                            \\
$d_u^n$    & $\mathrm{U}[100,300]m$ & $c_0$                & $2.7\times10^{8}$ cycles/sample                      \\ \bottomrule
\end{tabular}
}
\end{table}

{To demonstrate the efficacy of our proposed method, we evaluate its performance against the following advanced schemes:
\begin{itemize}
\item FedSGD \cite{mcmahan2017communication}: In each iteration, the edge server transmits the current model to all devices. Subsequently, the devices conduct model training in parallel and ultimately transmit their newly updated gradients back to the edge server for aggregation and global model updating. 
\item SignSGD \cite{pmlr-v80-bernstein18a}: In each iteration, the devices calculate the average gradient based on all data samples but only sent the sign of each coordinate of the gradient to the edge server. 
\item FedMP \cite{10050151}: In each iteration, the edge server formulates the pruning rate selection as a multi-armed bandit (MAB) problem to determine the optimal pruning rate for each device and sends a pruned model to each device. Devices then train the pruned model in parallel  and send the updated model to the edge server for consequent aggregation and global model updating. 
\item STC \cite{STC}: Based on FedSGD, this scheme compresses both uplink and downlink communications using sparsification, ternarization, error accumulation, and optimal Golomb encoding. 
\end{itemize}
For fair comparison, the transmission power of each device in these baselines is set to half of the $ p^{\rm max}$.}

\subsection{Experimental Analysis}
\subsubsection{Ablation Experiments.}

\begin{figure*}[t]
  \centering
  \setlength{\subfigcapskip}{5pt}
  \setlength{\belowcaptionskip}{5pt}
  \subfigure[Convergence comparison.]{
  \begin{minipage}{0.31\linewidth}\label{ComparisonAblation_Convergence}
  \centering
  \includegraphics[height=4.3cm, width=5.8cm]{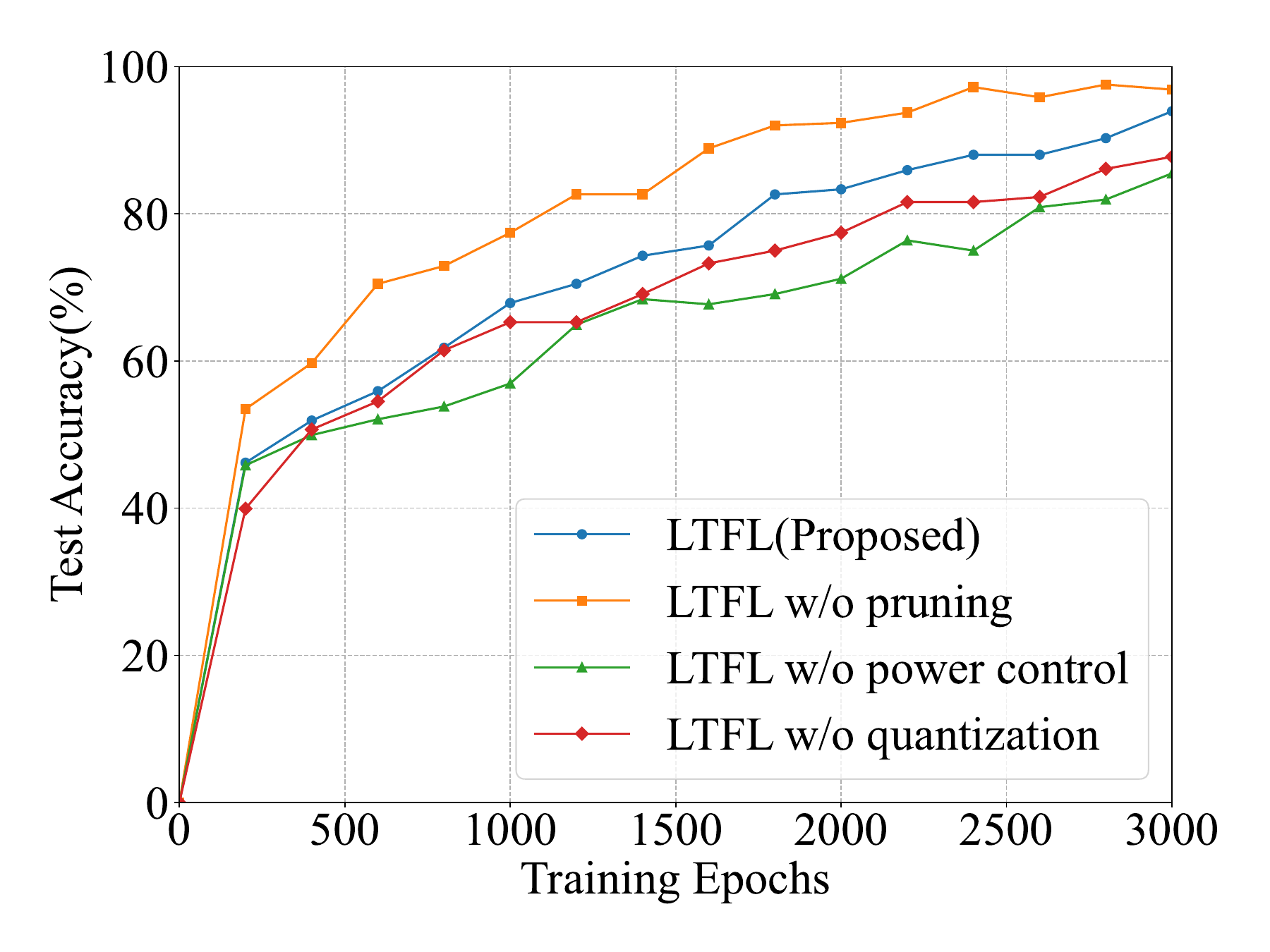}
  \end{minipage}
  }
  \subfigure[Delay comparison.]{
  \begin{minipage}{0.31\linewidth}\label{ComparisonAblation_Delay}
  \centering
  \includegraphics[height=4.3cm, width=5.8cm]{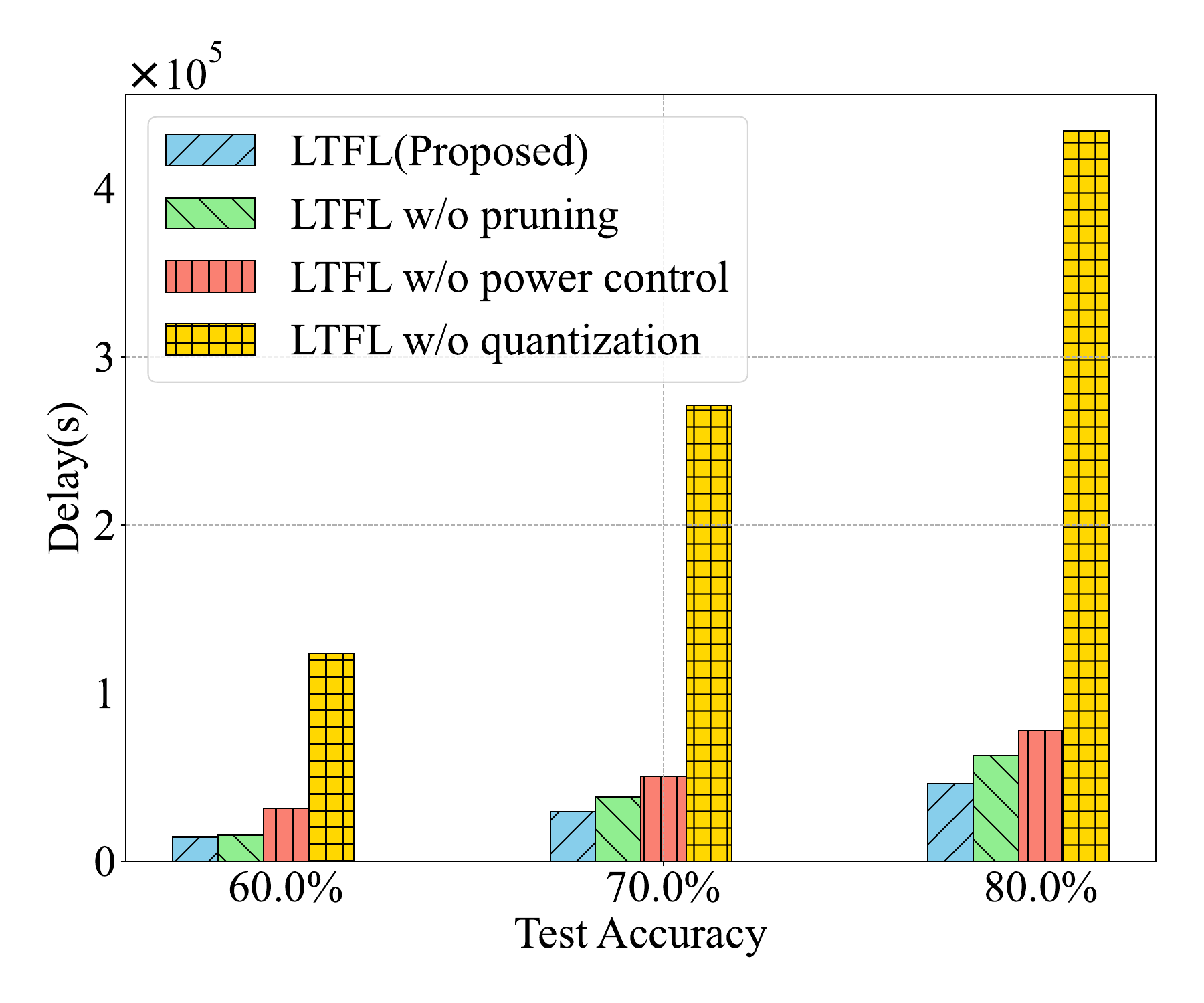}
  \end{minipage}
  }
  \subfigure[Energy consumption comparison.]{
  \begin{minipage}{0.31\linewidth}\label{ComparisonAblation_EnergyConsumption}
  \centering
  \includegraphics[height=4.3cm, width=5.8cm]{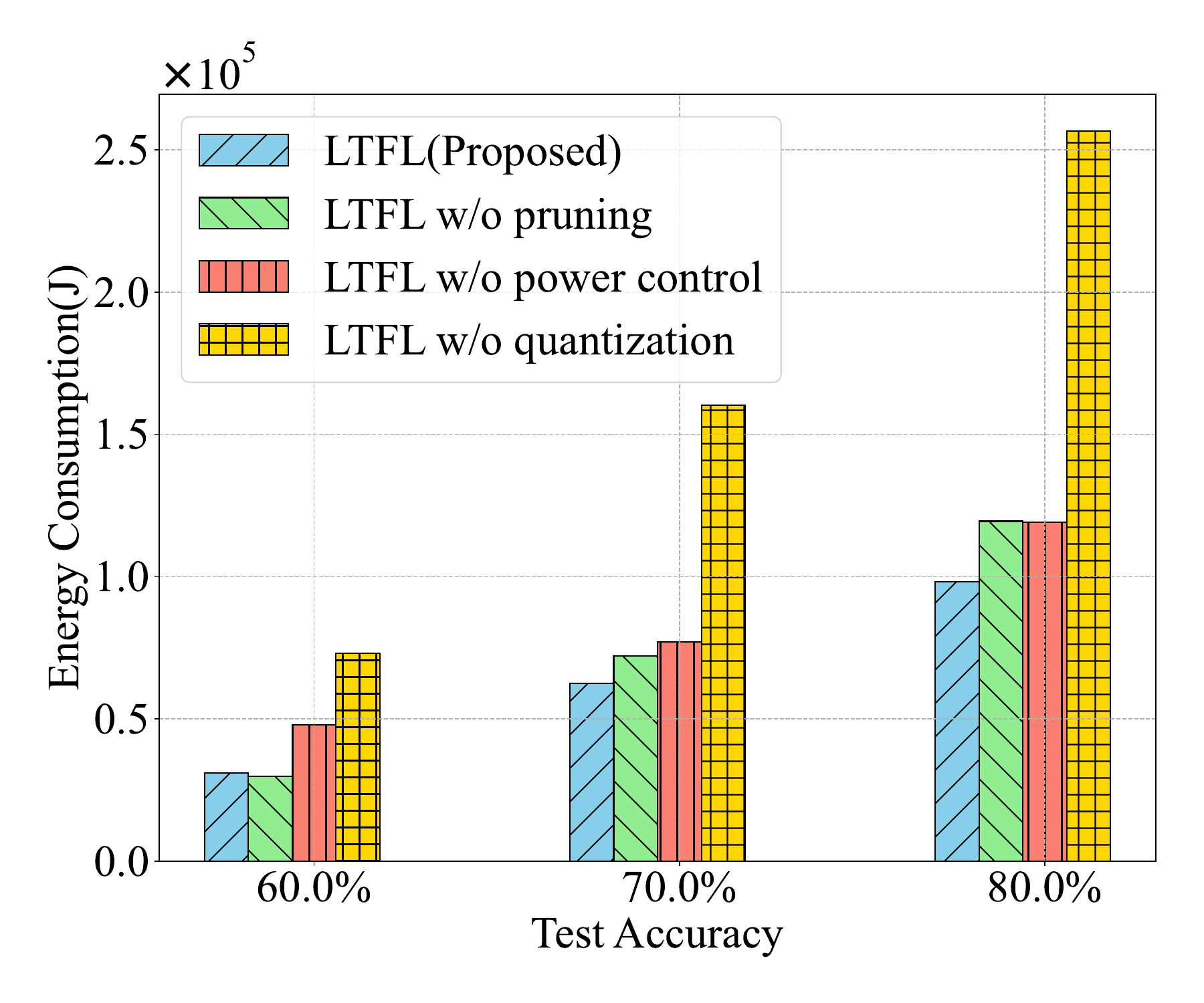}
  \end{minipage}
  }
  \centering
  \caption {Ablation experiments.}\label{ComparisonAblation}
\end{figure*}

{To evaluate the effectiveness of LTFL, we compare its performance with three ablated variants that exclude specific optimization techniques: (i) LTFL without pruning, (ii) LTFL without quantization, and (iii) LTFL without power control. As shown in Fig.  \ref{ComparisonAblation}, the variant without pruning achieves the fastest convergence, as it bypasses model compression. However, this comes at the cost of higher communication and energy overhead. The variant without quantization converges at a rate similar to LTFL but incurs the highest overall cost due to the large model size and resulting communication burden. In contrast, the variant without power control exhibits the worst convergence behavior, as the lack of adaptive power adjustment leads to frequent device dropouts, severely degrading system performance.
LTFL, while converging slightly more slowly than the no-pruning variant, ultimately achieves comparable accuracy. More importantly, it significantly reduces both time and energy overhead compared to all other variants. These results highlight the effectiveness of integrating pruning, quantization, and power control to balance convergence speed, accuracy, and resource efficiency.}


\begin{figure*}[t]
  \centering
  \setlength{\subfigcapskip}{5pt}
  \setlength{\belowcaptionskip}{5pt}
  \subfigure[Convergence comparison.]{
  \begin{minipage}{0.31\linewidth}\label{ComparisonDifferentScheme_Convergence}
  \centering
  \includegraphics[height=4.3cm, width=5.8cm]{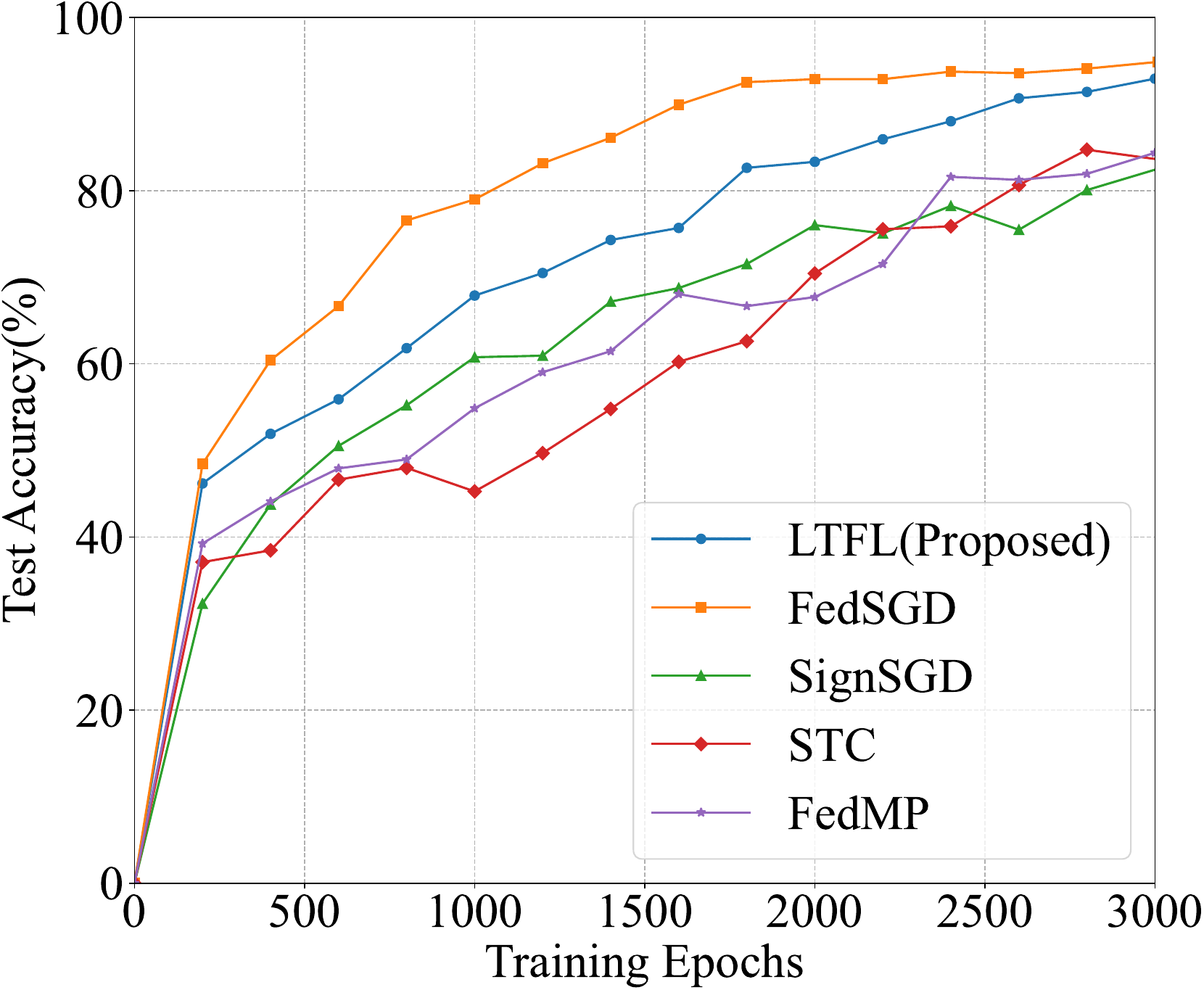}
  \end{minipage}
  }
  \subfigure[Delay comparison.]{
  \begin{minipage}{0.31\linewidth}\label{ComparisonDifferentScheme_Delay}
  \centering
  \includegraphics[height=4.3cm, width=5.8cm]{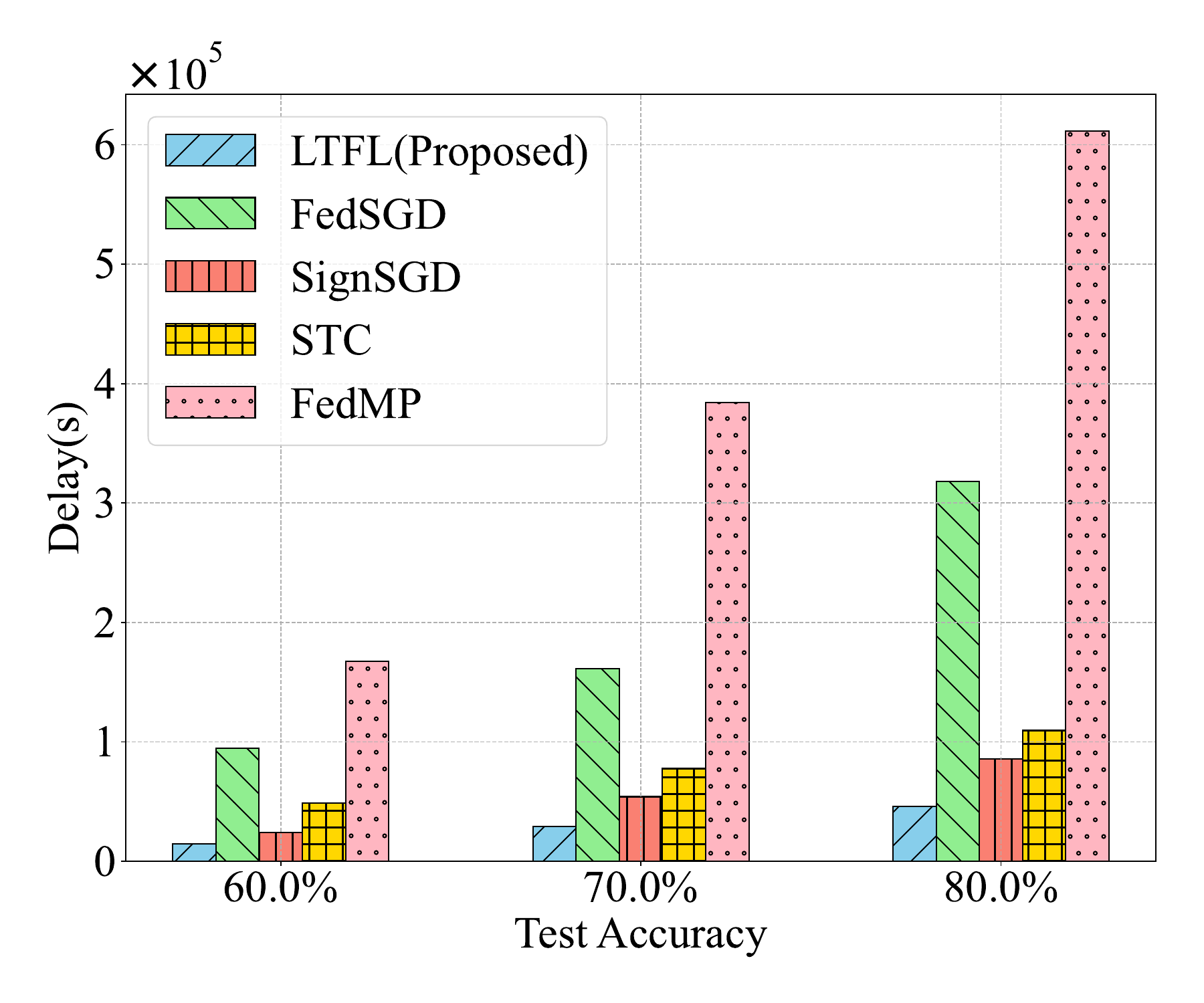}
  \end{minipage}
  }
  \subfigure[Energy consumption comparison.]{
  \begin{minipage}{0.31\linewidth}\label{ComparisonDifferentScheme_EngergyConsumption}
  \centering
  \includegraphics[height=4.3cm, width=5.8cm]{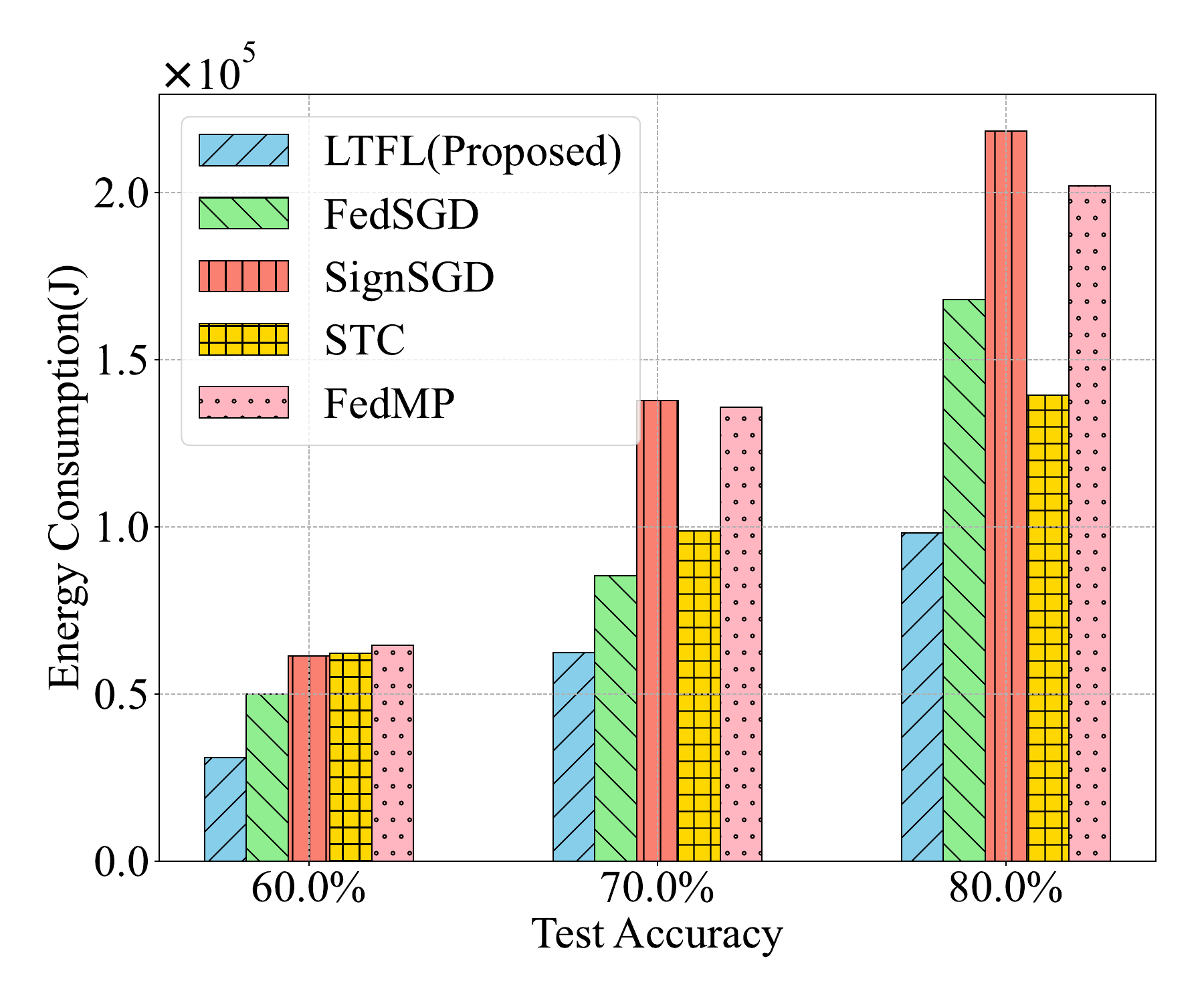}
  \end{minipage}
  }
  \centering
  \caption {Comparison of different schemes.}\label{ComparisonDifferentScheme}
\end{figure*}

\subsubsection{Comparison of Different FL Schemes}
{To demonstrate that the proposed LTFL can markedly reduce both time overhead and energy consumption while maintaining commendable accuracy, we benchmark the convergence performance and training cost (including delay and energy consumption) of various FL schemes through a series of comparisons, as shown in Fig. \ref{ComparisonDifferentScheme_Convergence}, Fig\ref{ComparisonDifferentScheme_Delay} and \ref{ComparisonDifferentScheme_EngergyConsumption}. Fig. \ref{ComparisonDifferentScheme_Convergence} illustrates that LTFL initially lags behind FedSGD; however, it ultimately converges to a level comparable to FedSGD while achieving significant reductions in energy consumption and delay, despite a slight performance drop that may stem from integrating model pruning and gradient quantization techniques. In contrast, although SignSGD, FedMP and STC conserve transmission resources, their inferior convergence performance results in higher cumulative energy consumption and delay. The superior convergence of LTFL is primarily due to its consideration of wireless transmission effects—an aspect neglected by the other schemes—which often leads to unreliable transmissions and device dropouts, thereby reducing the training data available for convergence. Further examination of Fig. \ref{ComparisonDifferentScheme_Delay} and Fig. \ref{ComparisonDifferentScheme_EngergyConsumption} confirms that LTFL excels in minimizing both time overhead and energy consumption to reach a given level of accuracy. Notably, while SignSGD requires transmission of only a minimal number of bits, its high communication loss incurs increased computational demands and more iterations, culminating in suboptimal delay and energy performance.}



\subsubsection{Impact of  The Channel Quality}

\begin{figure*}[t]
  \centering
  \setlength{\subfigcapskip}{5pt}
  \setlength{\belowcaptionskip}{5pt}
  \subfigure[Convergence comparison (Poor channel conditions).]{
  \begin{minipage}{0.32\linewidth}\label{ComparisonChannel_w=001_Convergence}
  \centering
  \includegraphics[height=4.3cm, width=5.8cm]{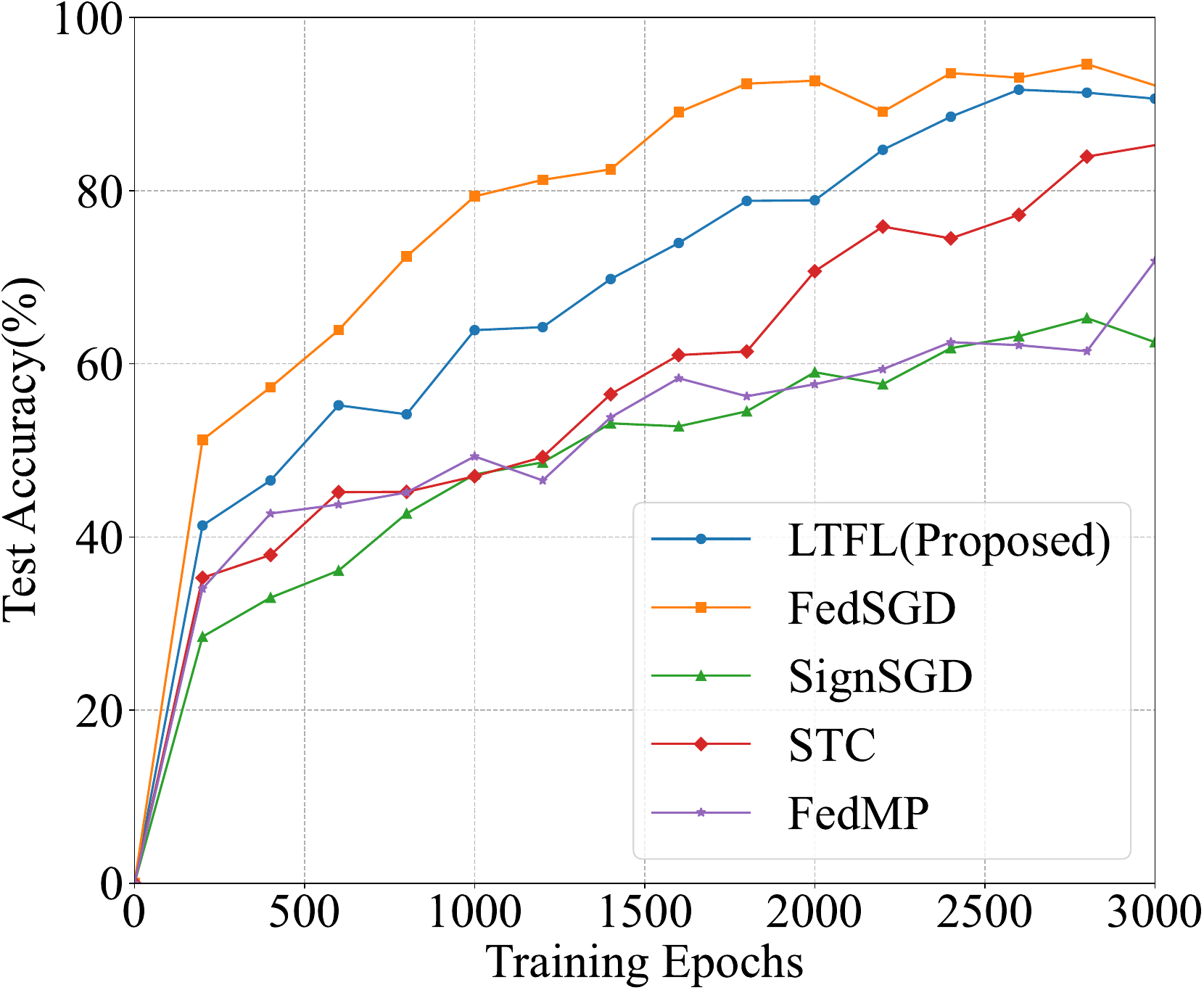}
  \end{minipage}
  }
  \subfigure[Convergence comparison (Normal channel conditions).]{
  \begin{minipage}{0.31\linewidth}\label{ComparisonChannel_w=002_Convergence}
  \centering
  \includegraphics[height=4.3cm, width=5.8cm]{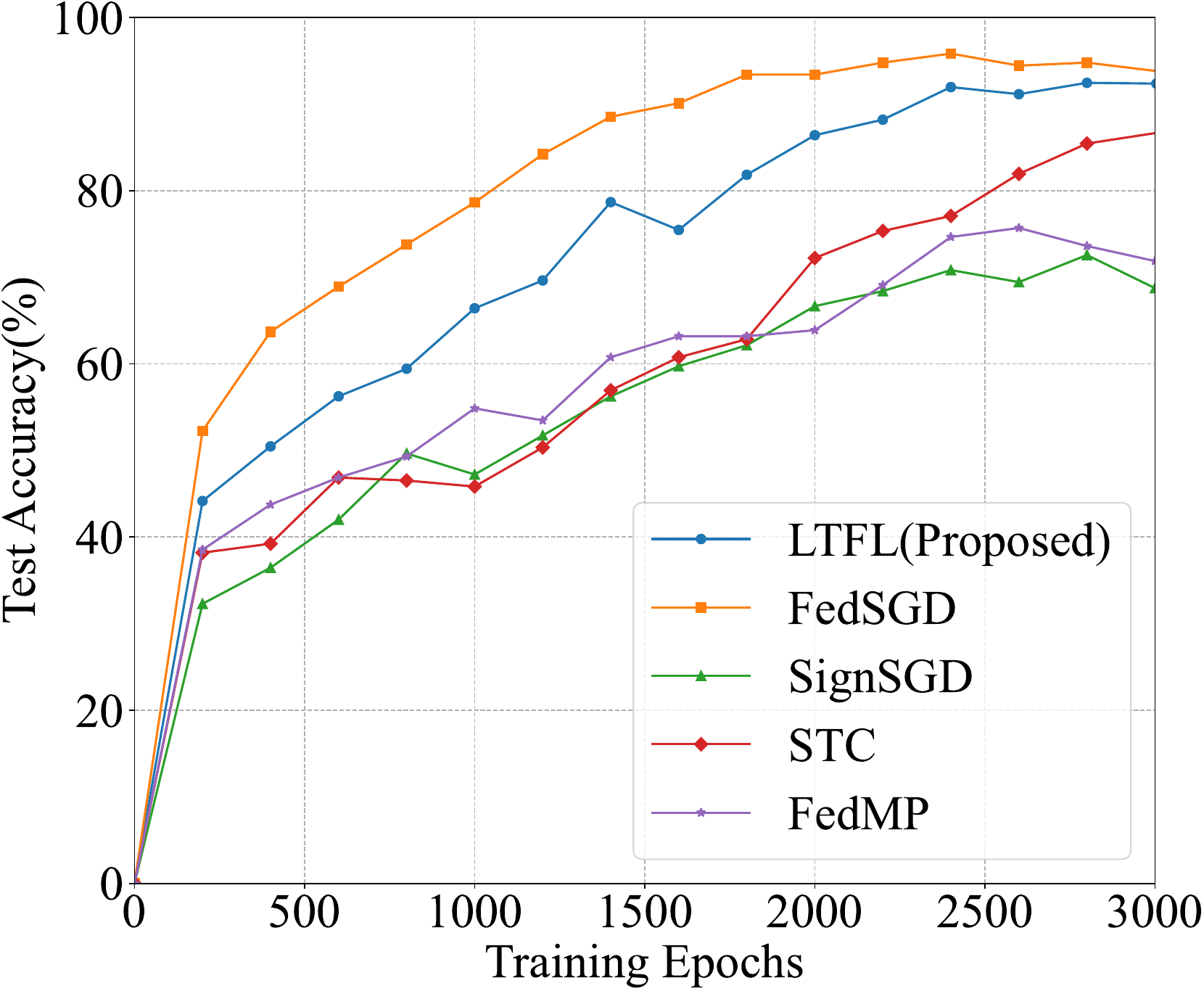}
  \end{minipage}
  }
  \subfigure[Convergence comparison (Good channel conditions).]{
  \begin{minipage}{0.31\linewidth}\label{ComparisonChannel_w=003_Convergence}
  \centering
  \includegraphics[height=4.3cm, width=5.8cm]{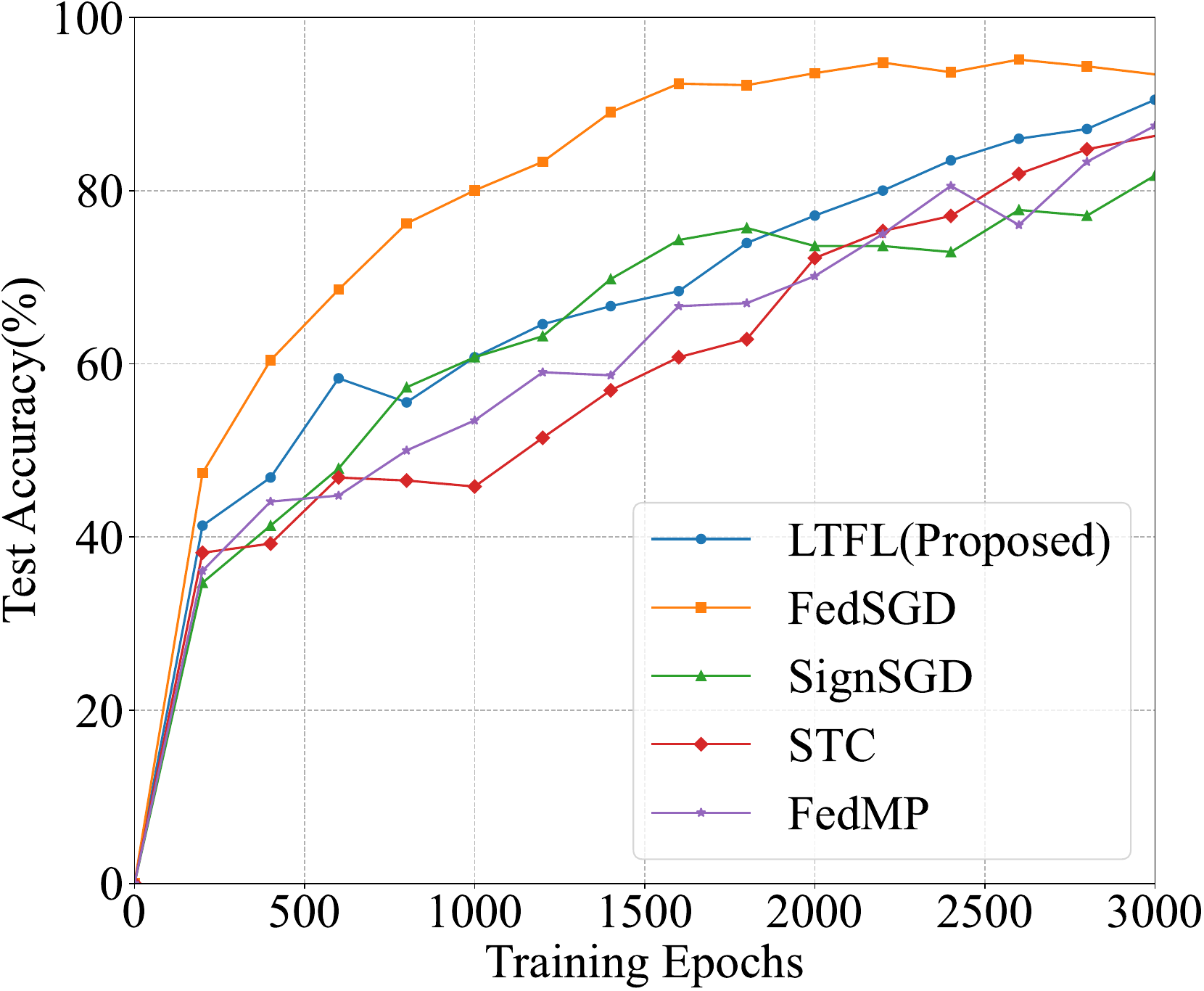}
  \end{minipage}
  }
  \centering
  \caption {Convergence comparison of different schemes under different channel conditions.}\label{ComparisonChannel_Convergence}
\end{figure*}

\begin{figure*}[t]
  \centering
  \setlength{\subfigcapskip}{5pt}
  \setlength{\belowcaptionskip}{5pt}
  \subfigure[Delay comparison (Poor channel conditions).]{
  \begin{minipage}{0.32\linewidth}\label{ComparisonChannel_w=001_Delay}
  \centering
  \includegraphics[height=4.3cm, width=5.8cm]{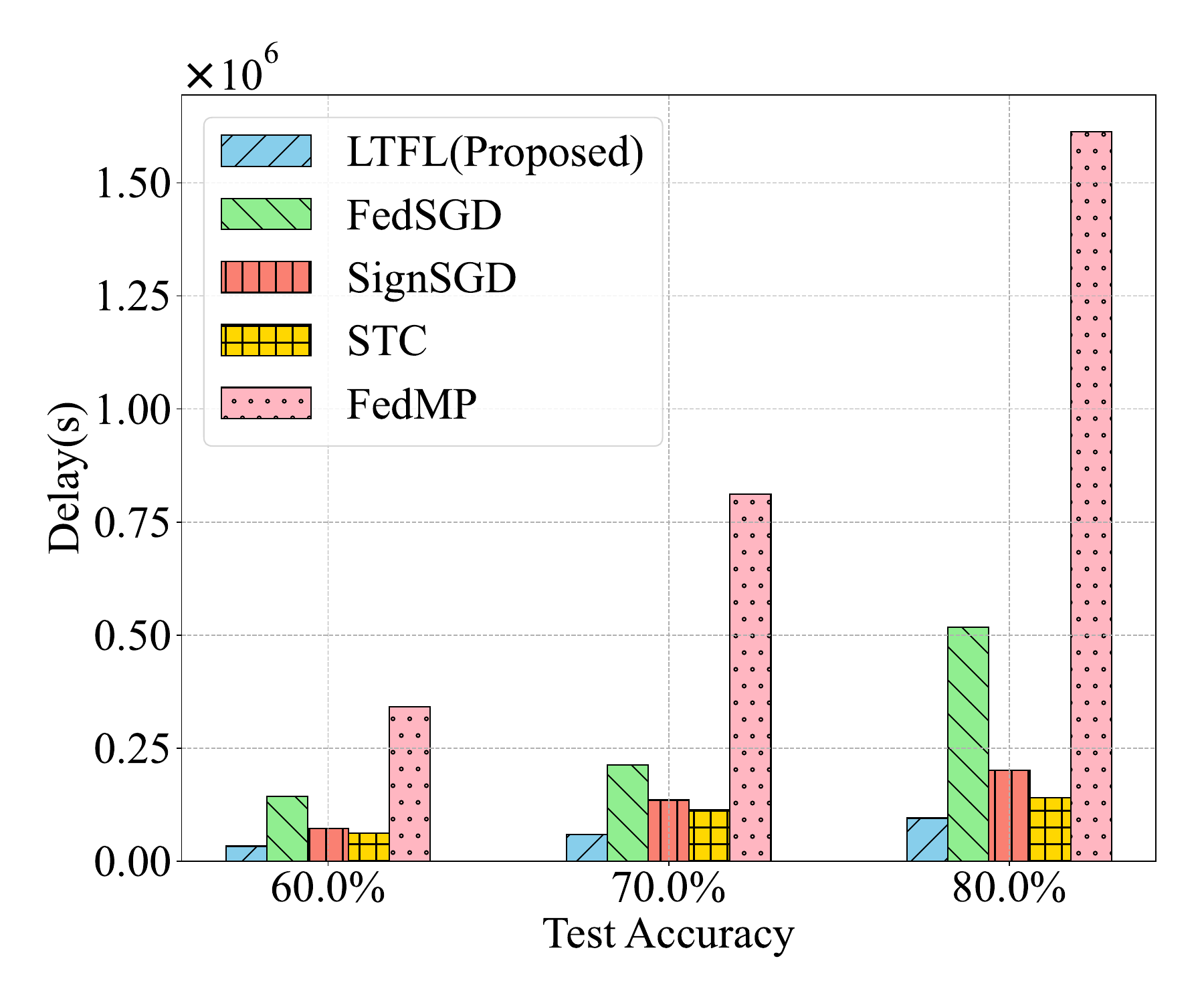}
  \end{minipage}
  }
  \subfigure[Delay comparison (Normal channel conditions).]{
  \begin{minipage}{0.31\linewidth}\label{ComparisonChannel_w=002_Delay}
  \centering
  \includegraphics[height=4.3cm, width=5.8cm]{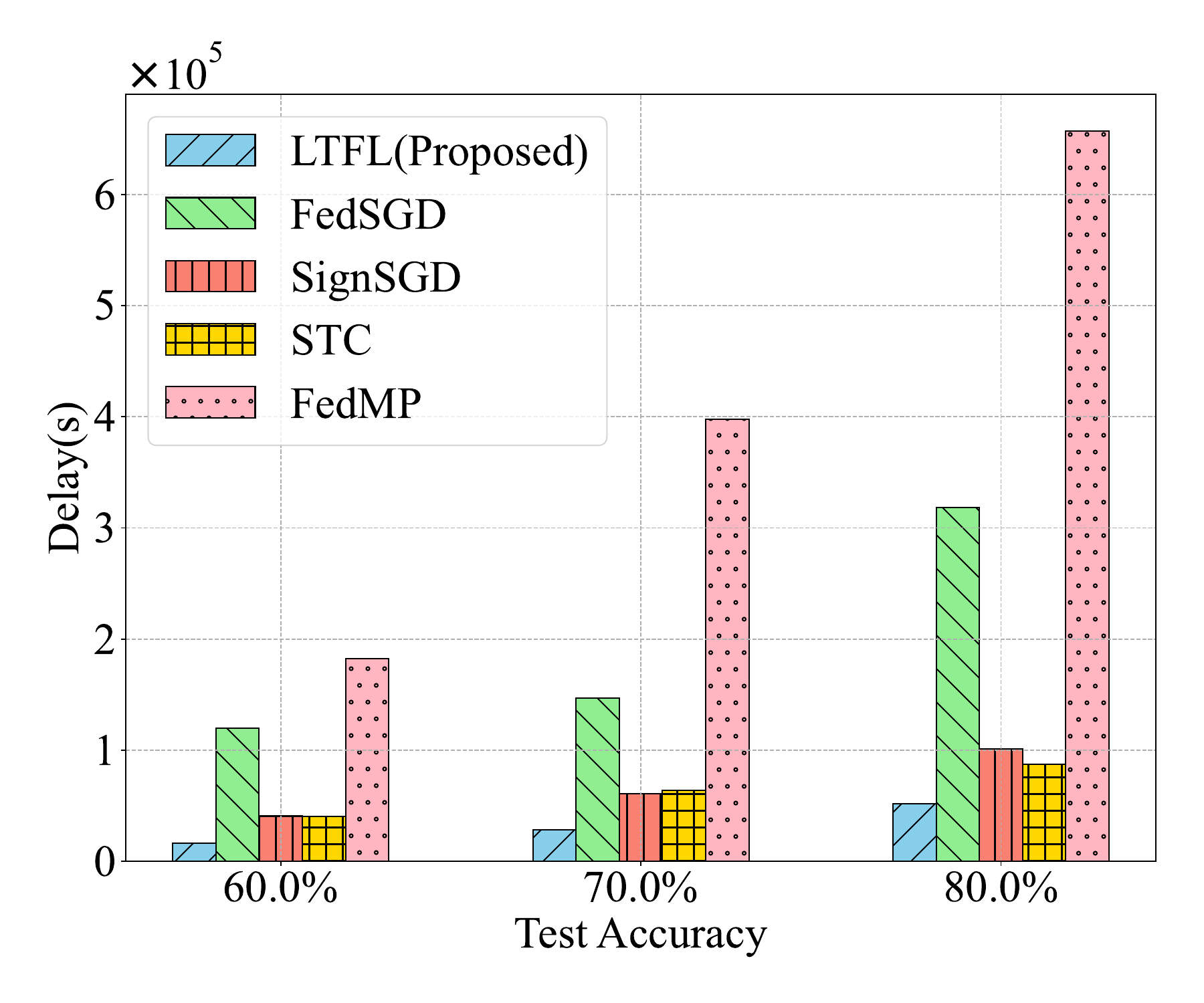}
  \end{minipage}
  }
  \subfigure[Delay comparison (Good channel conditions).]{
  \begin{minipage}{0.31\linewidth}\label{ComparisonChannel_w=003_Delay}
  \centering
  \includegraphics[height=4.3cm, width=5.8cm]{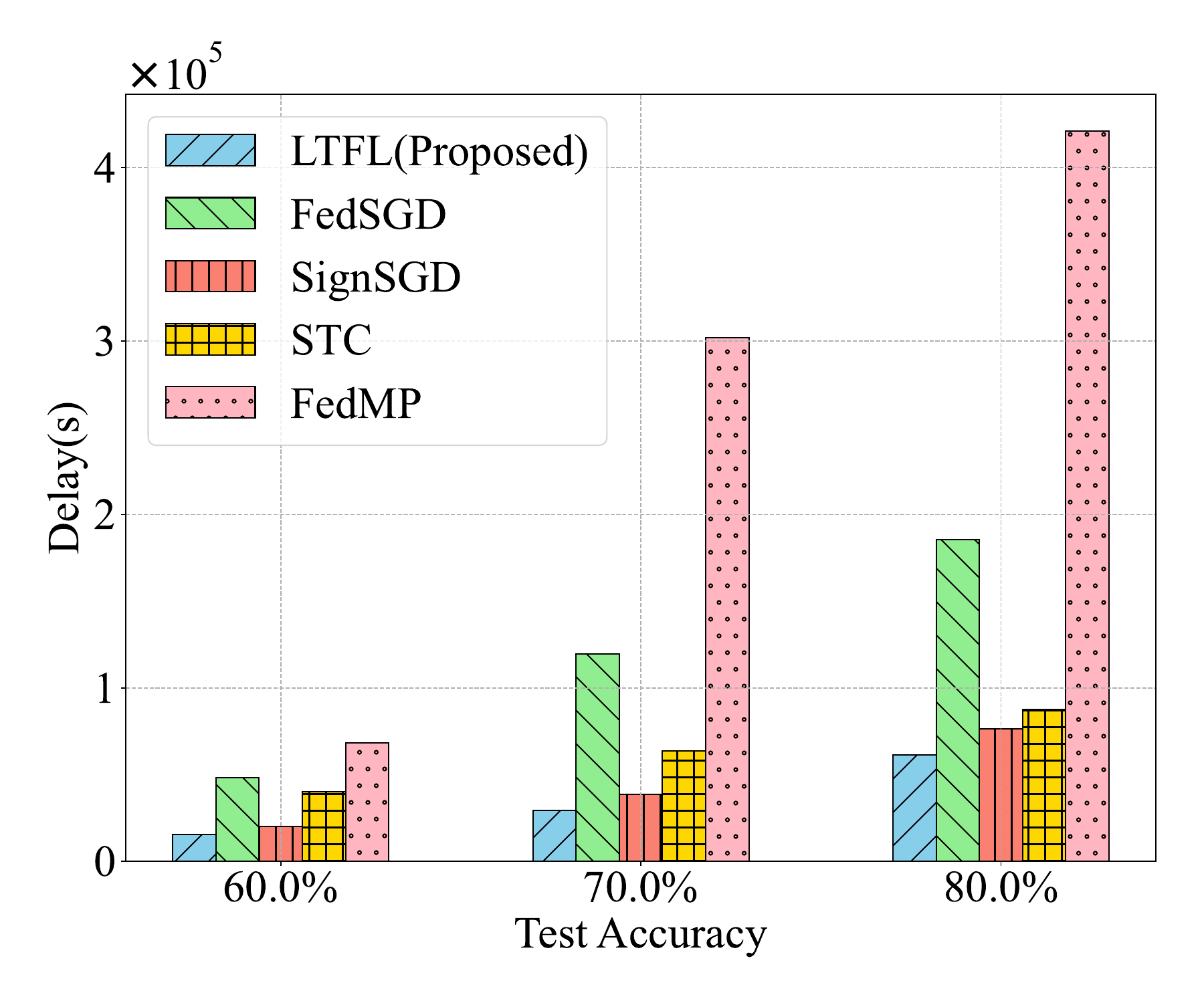}
  \end{minipage}
  }
  \centering
  \caption {Delay comparison of different schemes under different channel conditions.}\label{ComparisonChannel_Delay}
\end{figure*}

\begin{figure*}[t]
  \centering
  \setlength{\subfigcapskip}{5pt}
  \setlength{\belowcaptionskip}{5pt}
  \subfigure[Energy consumption comparison (Poor channel conditions).]{
  \begin{minipage}{0.31\linewidth}\label{ComparisonChannel_w=001_EnergyConsumption}
  \centering
  \includegraphics[height=4.3cm, width=5.8cm]{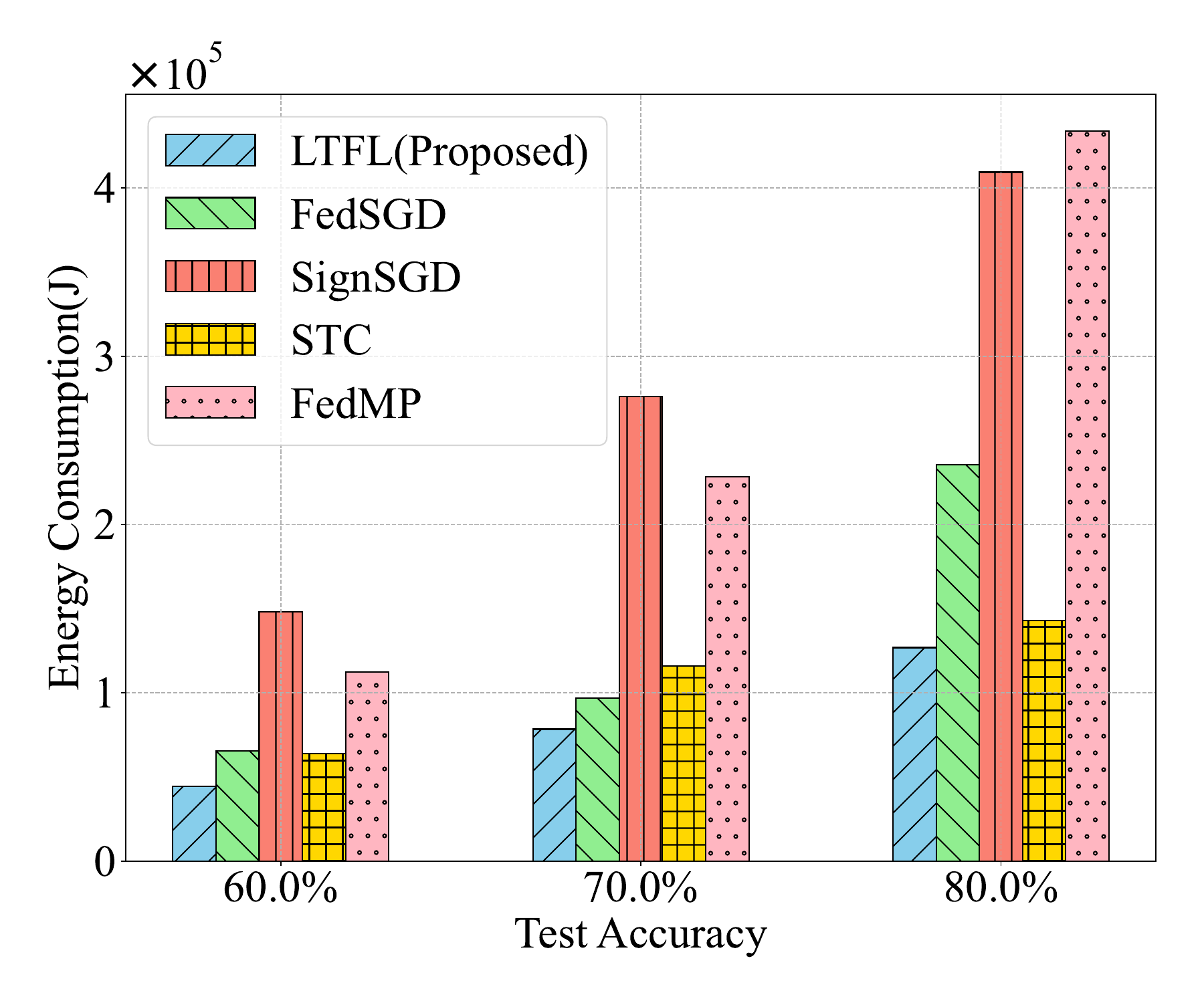}
  \end{minipage}
  }
  \subfigure[Energy consumption comparison (Normal channel conditions).]{
  \begin{minipage}{0.31\linewidth}\label{ComparisonChannel_w=002_EnergyConsumption}
  \centering
  \includegraphics[height=4.3cm, width=5.8cm]{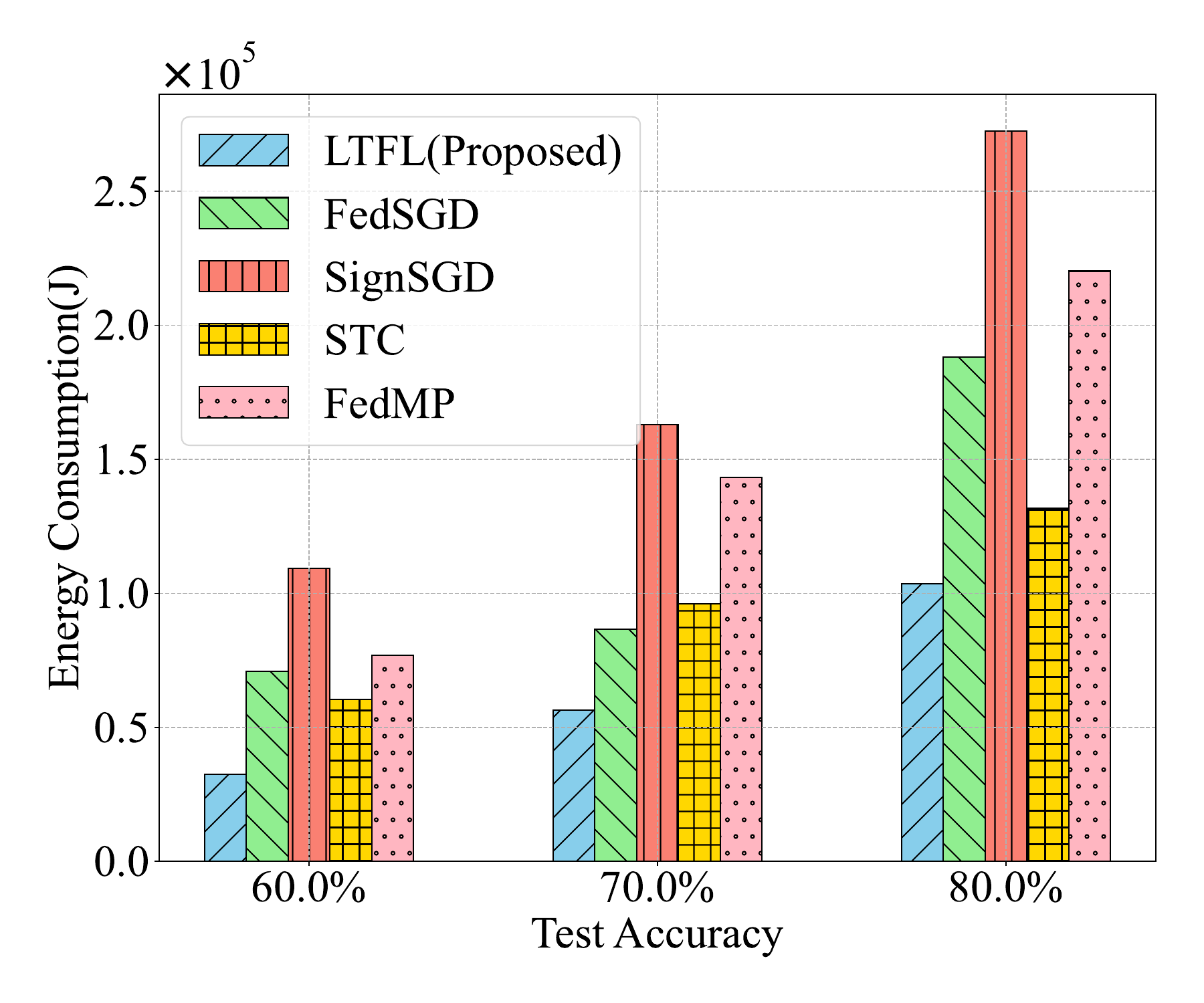}
  \end{minipage}
  }
  \subfigure[Energy consumption comparison (Good channel conditions).]{
  \begin{minipage}{0.31\linewidth}\label{ComparisonChannel_w=003_EnergyConsumption}
  \centering
  \includegraphics[height=4.3cm, width=5.8cm]{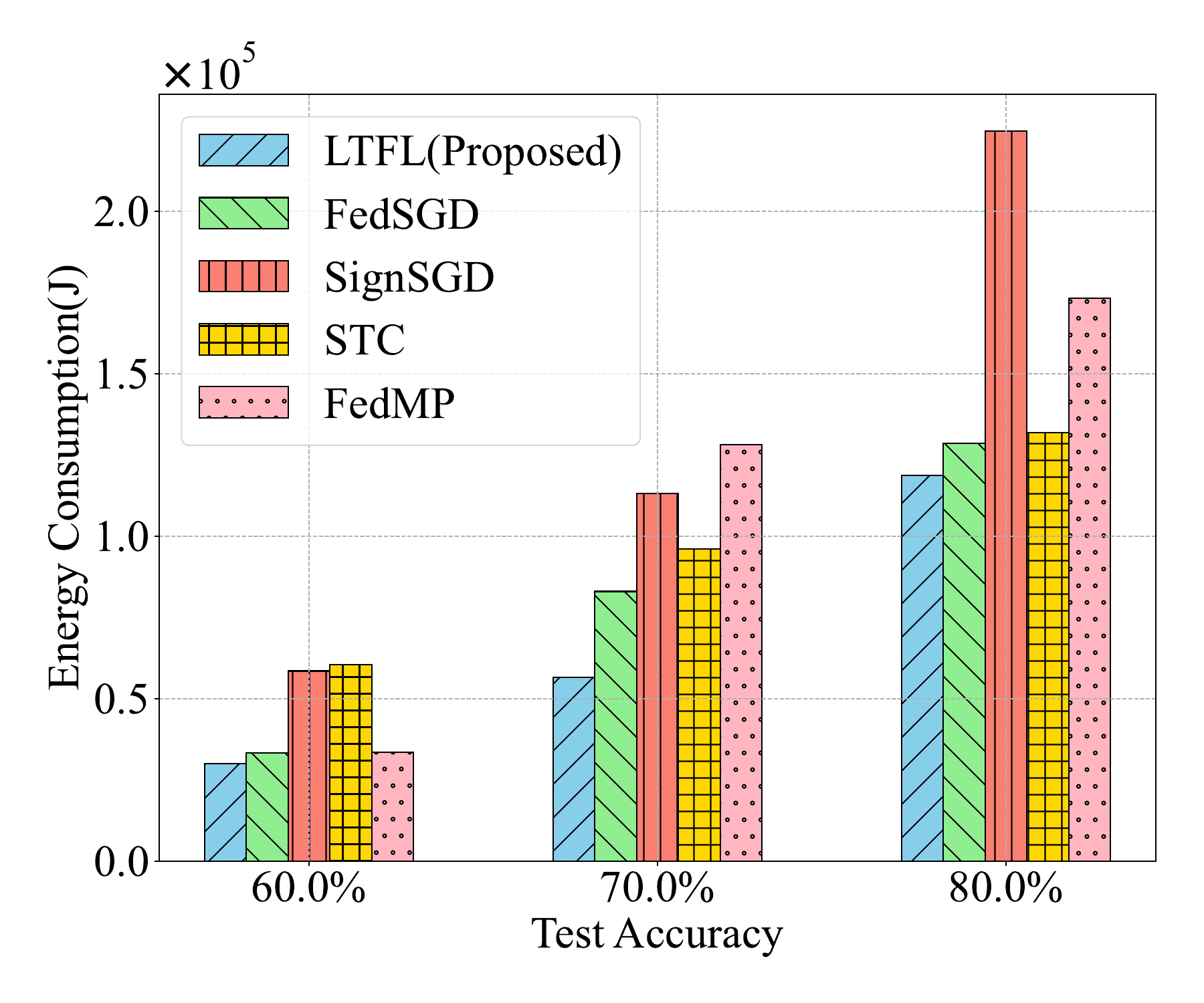}
  \end{minipage}
  }
  \centering
  \caption {Energy consumption comparison of different schemes under different channel conditions.}\label{ComparisonChannel_EnergyConsumption}
\end{figure*}

{Fig. \ref{ComparisonChannel_Convergence}, Fig. \ref{ComparisonChannel_Delay} and Fig. \ref{ComparisonChannel_EnergyConsumption} effectively illustrates the adaptability of the proposed LTFL framework in various wireless environments by setting distinct values of $\varpi^{n}_u$. Specifically, $\varpi^{n}_u = 0.01$ is indicative of poor channel quality, whereas $\varpi^{n}_u = 0.02$ signifies normal channel quality, and $\varpi^{n}_u = 0.03$ signifies good channel quality. Observations from Figs. \ref{ComparisonChannel_Convergence} reveal a direct correlation between channel quality and convergence performance. Poorer channel quality leads to a higher packet error rate, which in turn causes more device outages and a reduction in the volume of available training data. Despite these challenges, LTFL consistently outperforms other schemes due to its ability to flexibly adjust transmission power, model pruning ratio, and quantization level. Furthermore, Fig. \ref{ComparisonChannel_Delay} and Fig. \ref{ComparisonChannel_EnergyConsumption} demonstrate LTFL's efficiency in achieving the desired accuracy with minimal training costs in different wireless conditions. This underscores its robust adaptability to varying wireless environments.}

\subsubsection{Impact of  The Number of Devices}
\begin{figure}[t]
\centering
\setlength{\subfigcapskip}{5pt}
\setlength{\belowcaptionskip}{5pt}
\subfigure[Delay comparison.]{
\begin{minipage}{0.8\linewidth}\label{ComparisonDeviceNumber_Delay}
\centering
\includegraphics[width=1\textwidth]{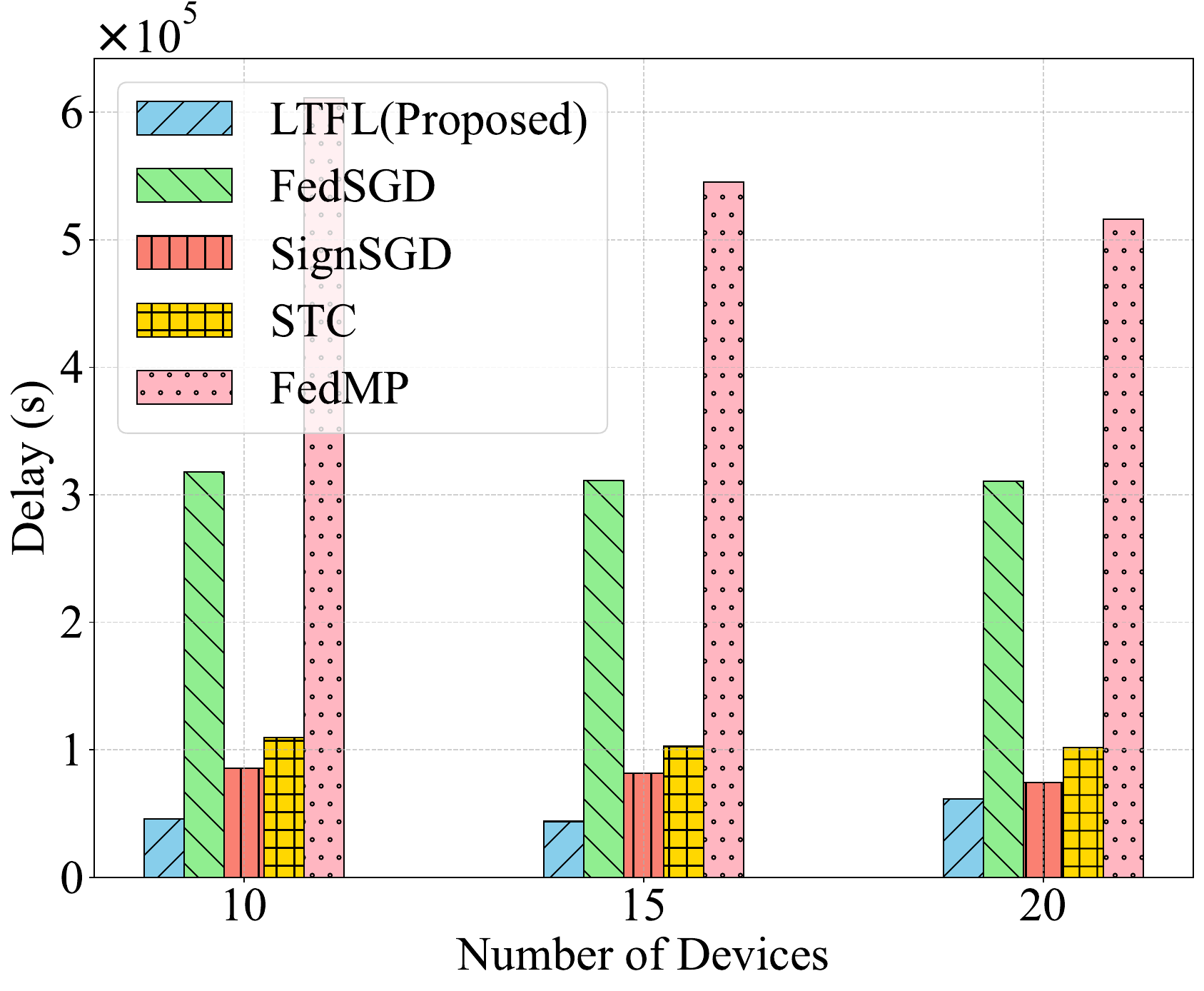}
\end{minipage}
}

\subfigure[Energy consumption comparison.]{
\begin{minipage}{0.8\linewidth}\label{ComparisonDeviceNumber_EnergyConsumption}
\centering
\includegraphics[width=\textwidth]{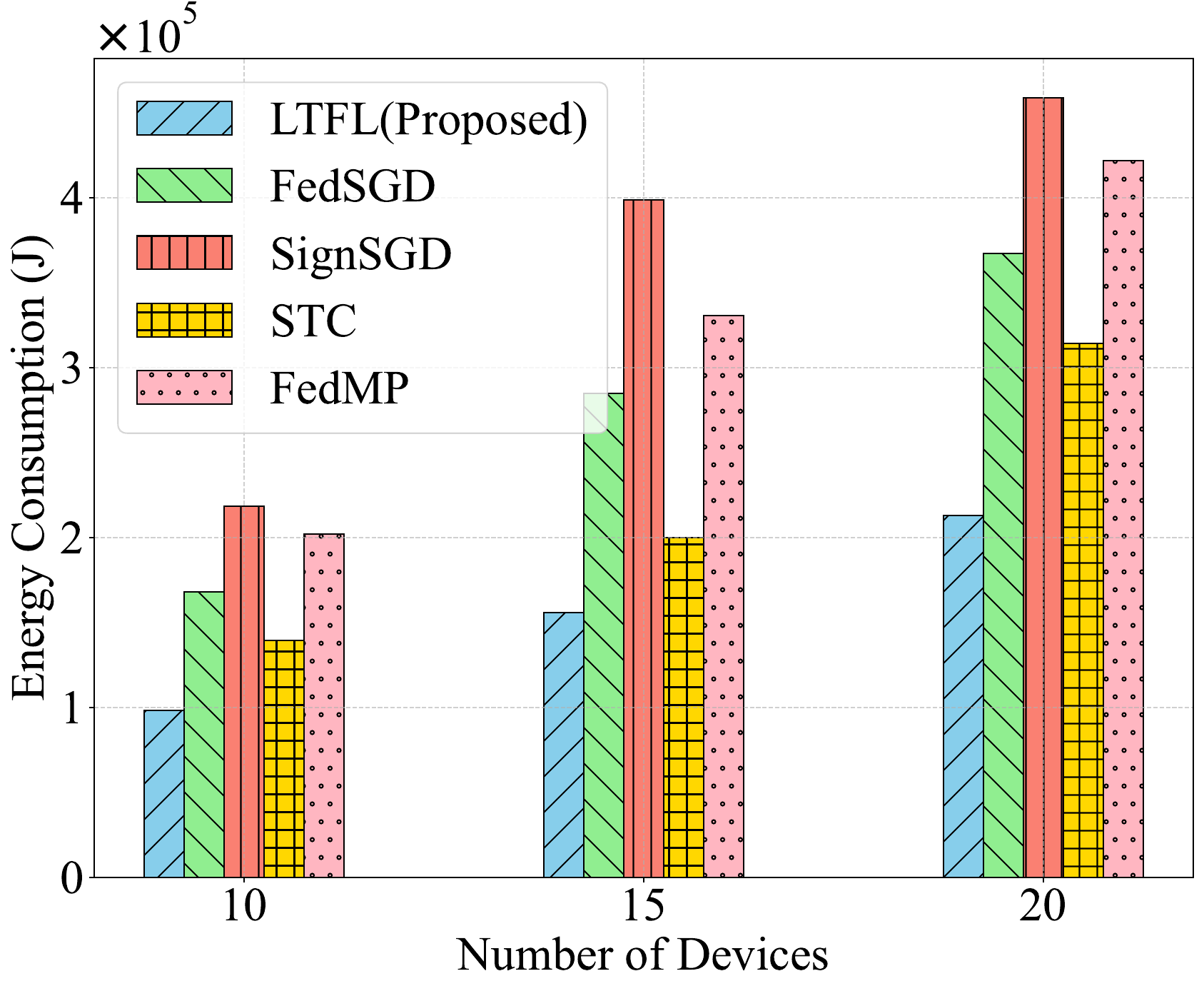}
\end{minipage}
}
\centering
\caption {Training cost comparison of different schemes versus the number of devices.}\label{ComparisonDeviceNumber}
\end{figure}

{Fig. \ref{ComparisonDeviceNumber} illustrates the impact of the number of devices on training cost. In our experiments, we set the number of devices to 10, 15, and 20, respectively. It becomes evident that an increase in the number of devices corresponds to lower training delay. The rationale behind this trend observed in Fig. \ref{ComparisonDeviceNumber_Delay} mirrors that in Fig. \ref{ComparisonChannel_Convergence} and Fig. \ref{ComparisonChannel_Delay}. Essentially, more devices mean a larger dataset for model training, which improves convergence performance, leading to fewer training epochs. As a result, the overall training delay is reduced. However, the energy consumption is not necessarily lower, as the impact of the reduced number of training epochs on energy consumption is offset by the increased number of devices, with each device contributing to the total energy consumption. Importantly, the LTFL framework consistently outperforms other schemes in terms of training cost efficiency, demonstrating its strong adaptability to variations in network scale.}

\subsubsection{Performance on Non-I.I.D Data}



{\color{black}To evaluate the performance of various baselines on non-i.i.d data, we adopt the Dirichlet distribution with a concentration parameter of $\alpha=0.1$, $\alpha=0.5$ and $\alpha=0.9$. This distribution is commonly used to simulate real-world data heterogeneity, where the data collected by each device is inherently imbalanced.
As shown in Fig. \ref{ComparisonNoniid_Convergence}, Fig. \ref{ComparisonNoniid_Delay} and \ref{ComparisonNoniid_EnergyConsumption}, all baselines experience a drop in accuracy due to the inherent data imbalance. Nevertheless, LTFL consistently achieves the lowest overhead while maintaining an accuracy level similar to FedSGD.
In conclusion, these results highlight LTFL's robustness in significantly reducing training overhead without compromising accuracy, even under challenging non-i.i.d conditions.}

\begin{figure*}[t]
  \centering
  \setlength{\subfigcapskip}{5pt}
  \setlength{\belowcaptionskip}{5pt}
  \subfigure[{\color{black}Convergence comparison (concentration parameter $\alpha=0.1$).}]{
  \begin{minipage}{0.31\linewidth}\label{ComparisonNoniid0.1_Convergence}
  \centering
  \includegraphics[height=4.3cm, width=5.8cm]{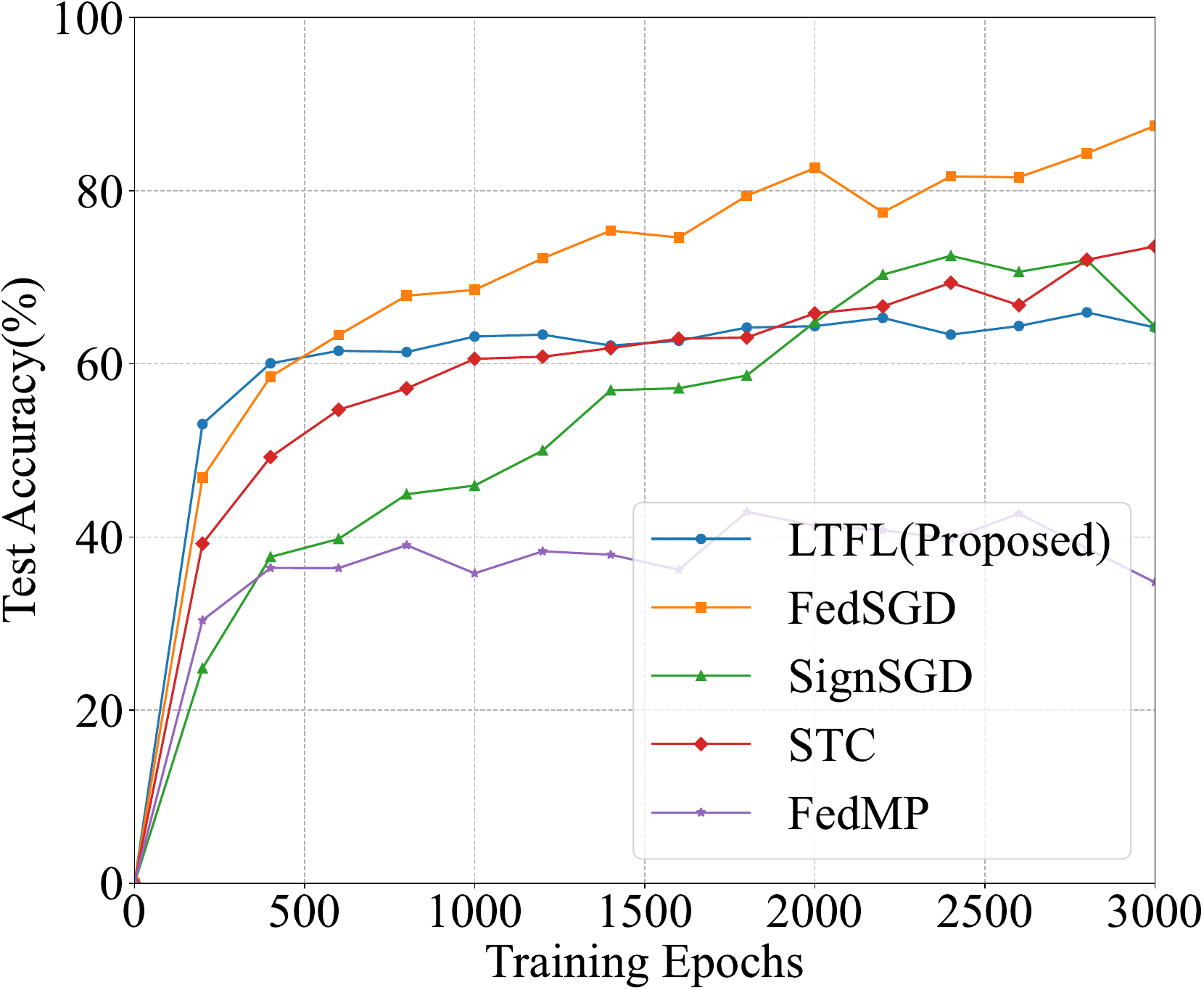}
  \end{minipage}
  }
  \subfigure[{\color{black}Convergence comparison (concentration parameter $\alpha=0.5$).}]{
  \begin{minipage}{0.31\linewidth}\label{ComparisonNoniid0.5_Convergence}
  \centering
  \includegraphics[height=4.3cm, width=5.8cm]{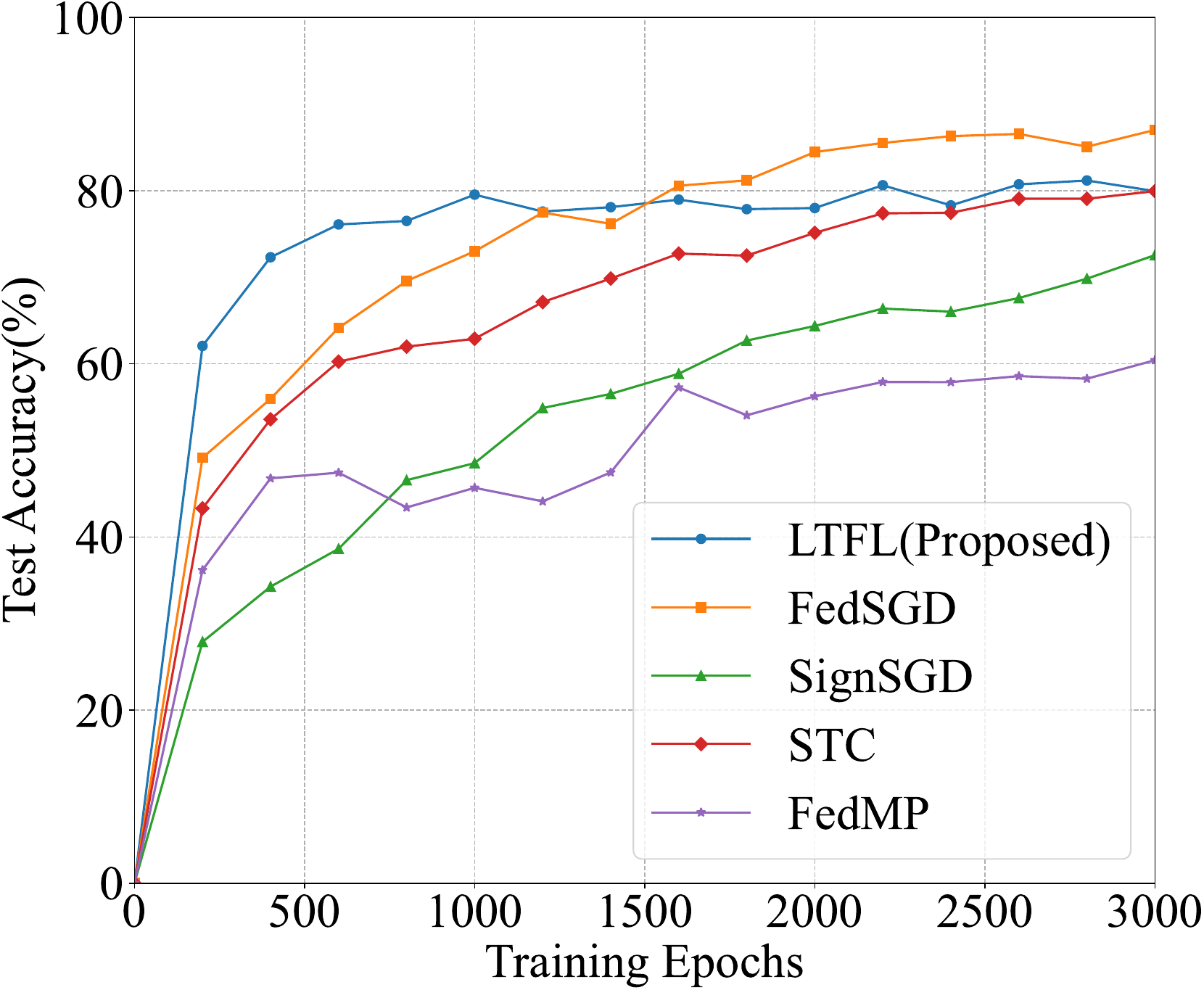}
  \end{minipage}
  }
  \subfigure[{\color{black}Convergence comparison (concentration parameter $\alpha=0.9$).}]{
  \begin{minipage}{0.31\linewidth}\label{ComparisonNoniid0.9_Convergence}
  \centering
  \includegraphics[height=4.3cm, width=5.8cm]{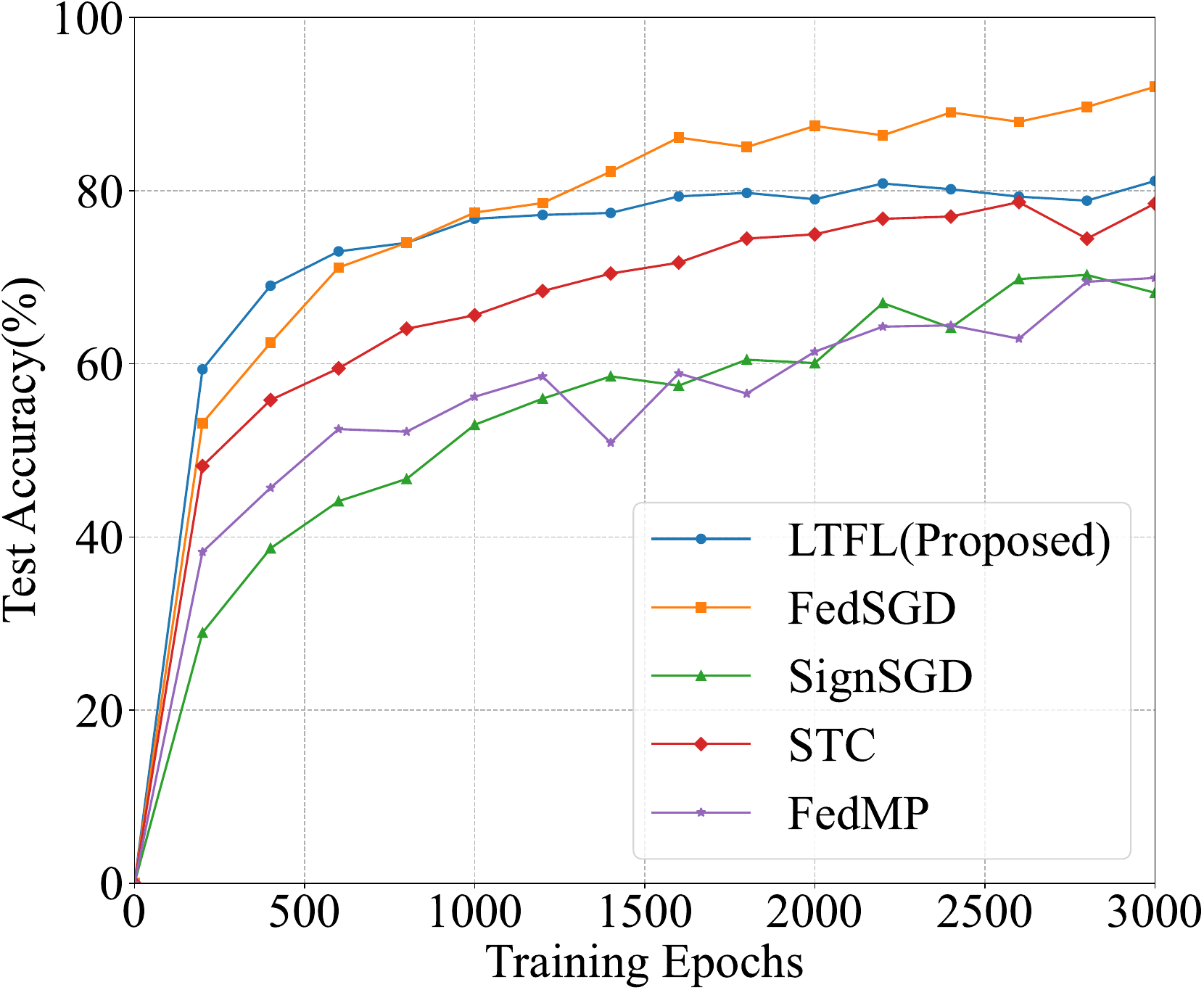}
  \end{minipage}
  }
  \centering
  \caption {{\color{black}Convergence comparison of different schemes under different non-i.i.d. scenarios.}}\label{ComparisonNoniid_Convergence}
\end{figure*}

\begin{figure*}[t]
  \centering
  \setlength{\subfigcapskip}{5pt}
  \setlength{\belowcaptionskip}{5pt}
  \subfigure[{\color{black}Delay comparison (concentration parameter $\alpha=0.1$).}]{
  \begin{minipage}{0.32\linewidth}\label{ComparisonNoniid0.1_Delay}
  \centering
  \includegraphics[height=4.3cm, width=5.8cm]{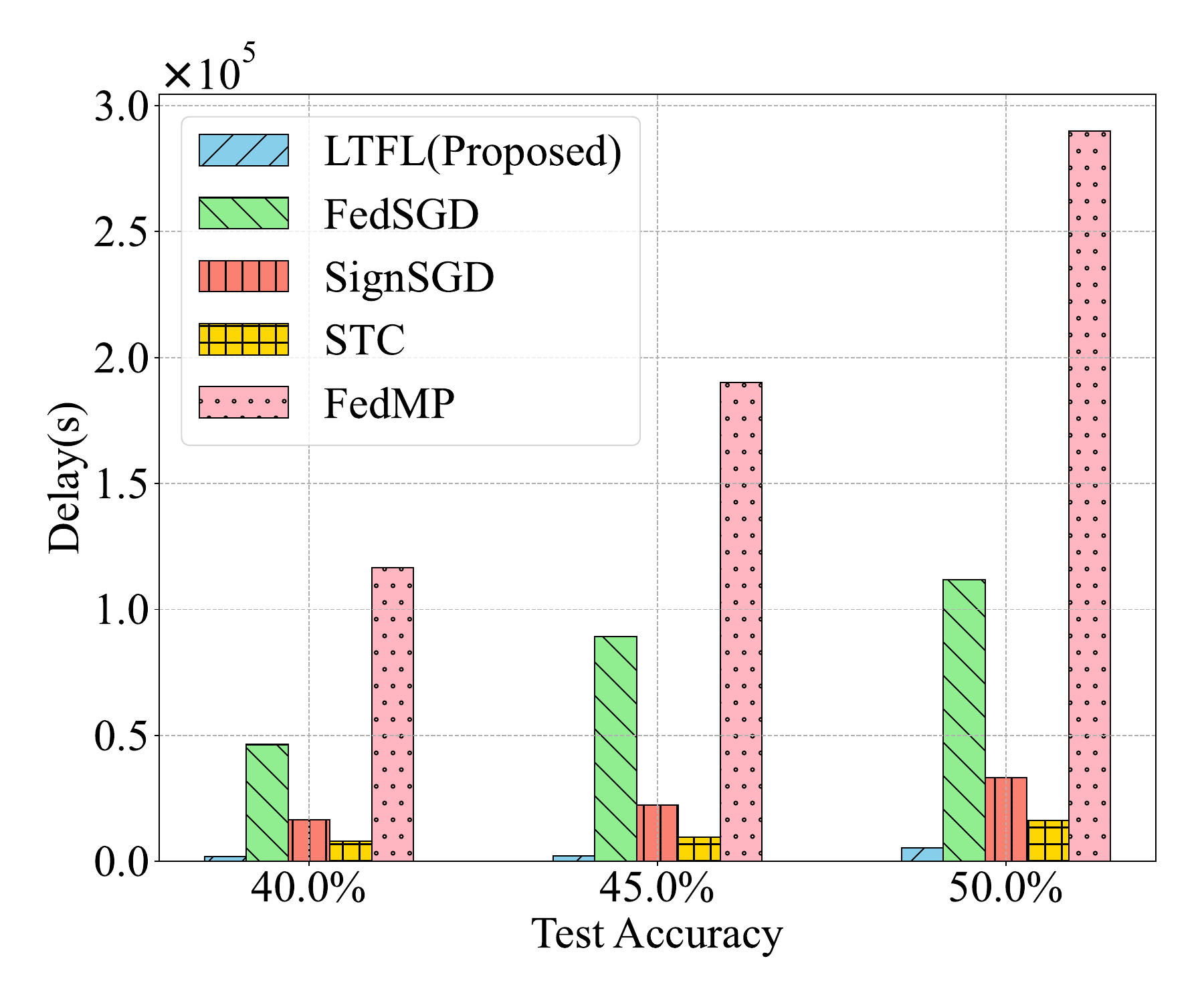}
  \end{minipage}
  }
  \subfigure[{\color{black}Delay comparison (concentration parameter $\alpha=0.5$).}]{
  \begin{minipage}{0.31\linewidth}\label{ComparisonNoniid0.5_Delay}
  \centering
  \includegraphics[height=4.3cm, width=5.8cm]{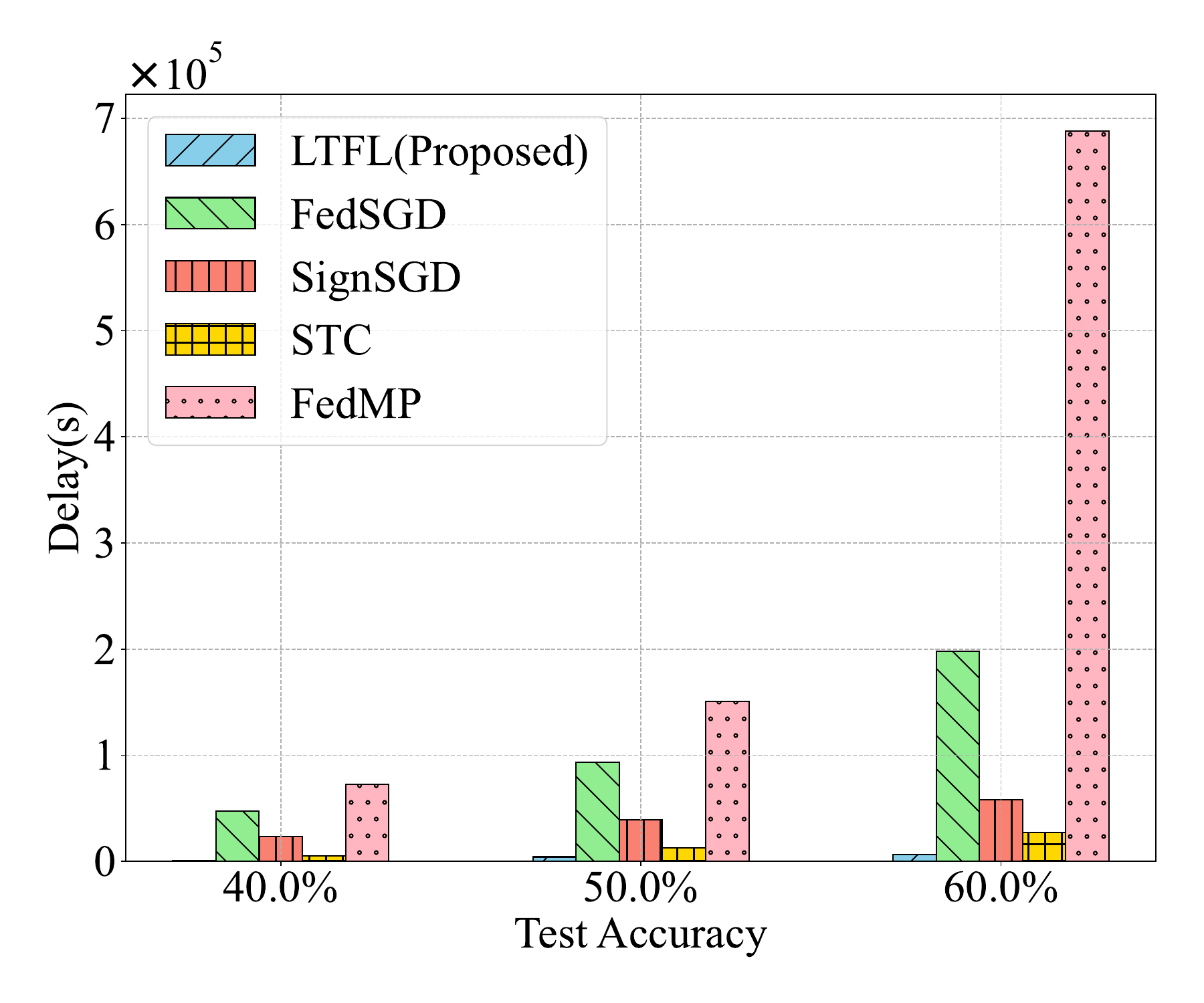}
  \end{minipage}
  }
  \subfigure[{\color{black}Delay comparison (concentration parameter $\alpha=0.9$).}]{
  \begin{minipage}{0.31\linewidth}\label{ComparisonNoniid0.5_Delay}
  \centering
  \includegraphics[height=4.3cm, width=5.8cm]{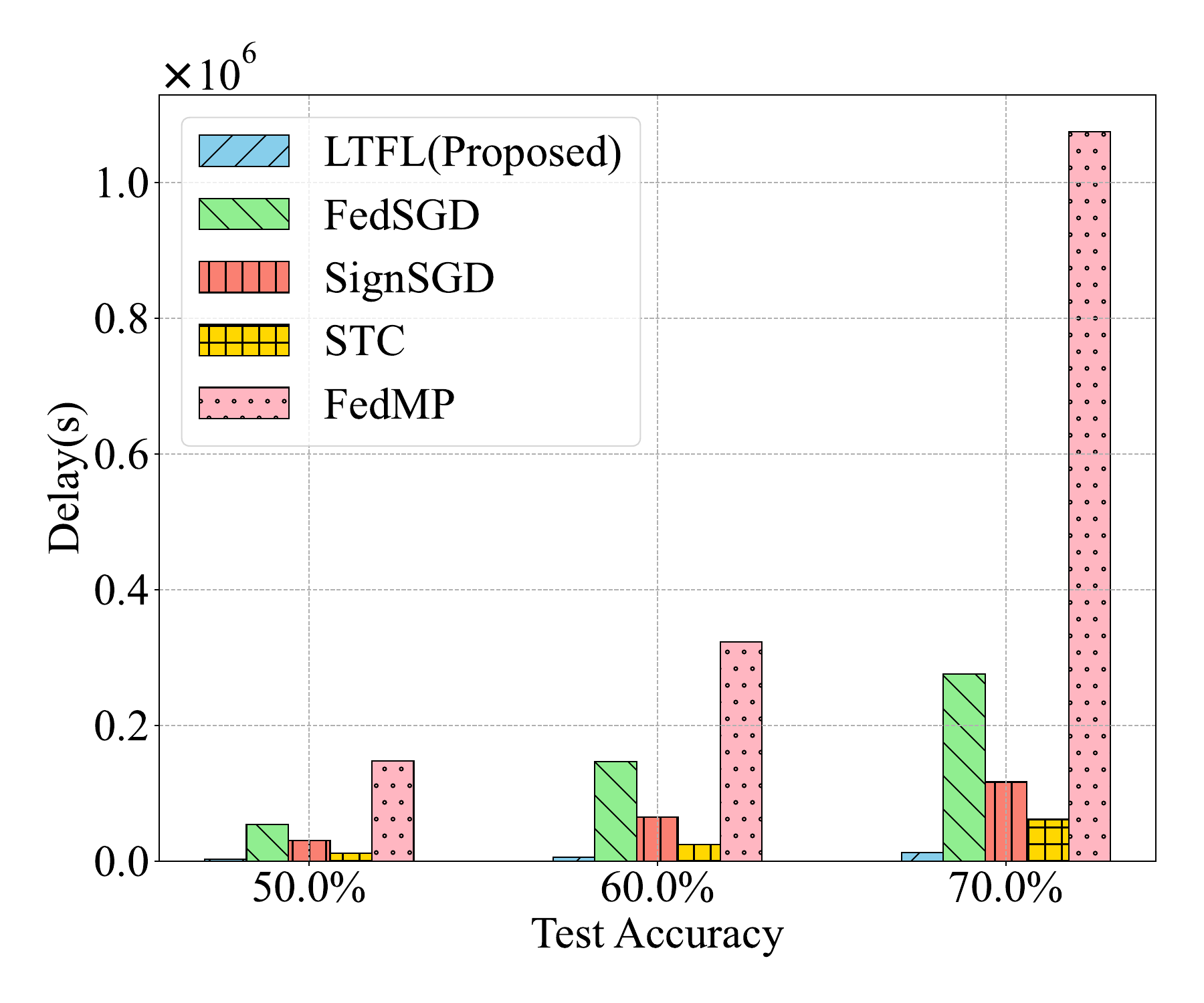}
  \end{minipage}
  }
  \centering
  \caption {{\color{black}Delay comparison of different schemes under different non-i.i.d. scenarios.}}\label{ComparisonNoniid_Delay}
\end{figure*}

\begin{figure*}[t]
  \centering
  \setlength{\subfigcapskip}{5pt}
  \setlength{\belowcaptionskip}{5pt}
  \subfigure[{\color{black}Energy consumption comparison (concentration parameter $\alpha=0.1$).}]{
  \begin{minipage}{0.32\linewidth}\label{ComparisonNoniid0.1_EnergyConsumption}
  \centering
  \includegraphics[height=4.3cm, width=5.8cm]{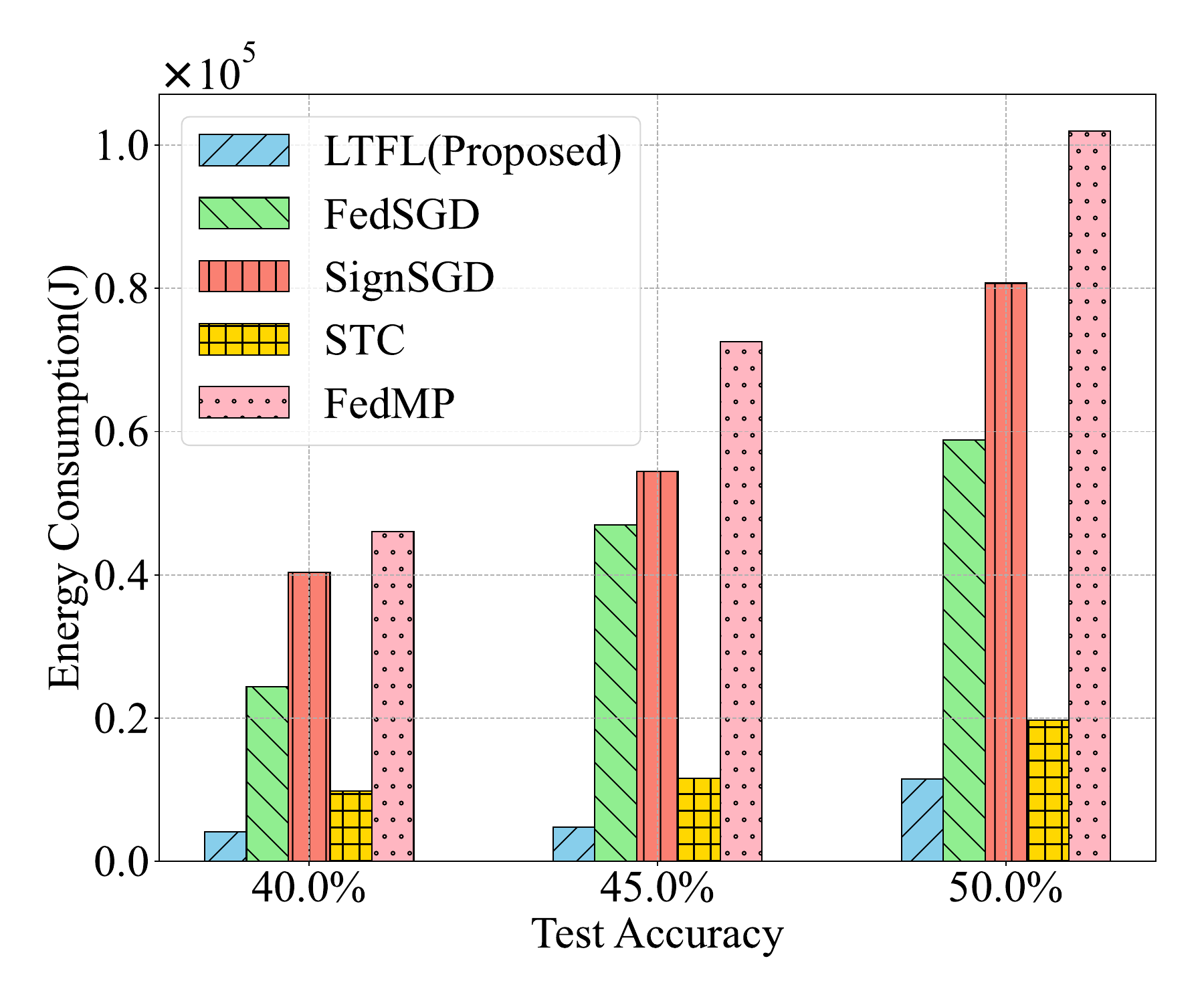}
  \end{minipage}
  }
  \subfigure[{\color{black}Energy consumption comparison (concentration parameter $\alpha=0.5$).}]{
  \begin{minipage}{0.31\linewidth}\label{ComparisonNoniid0.5_EnergyConsumption}
  \centering
  \includegraphics[height=4.3cm, width=5.8cm]{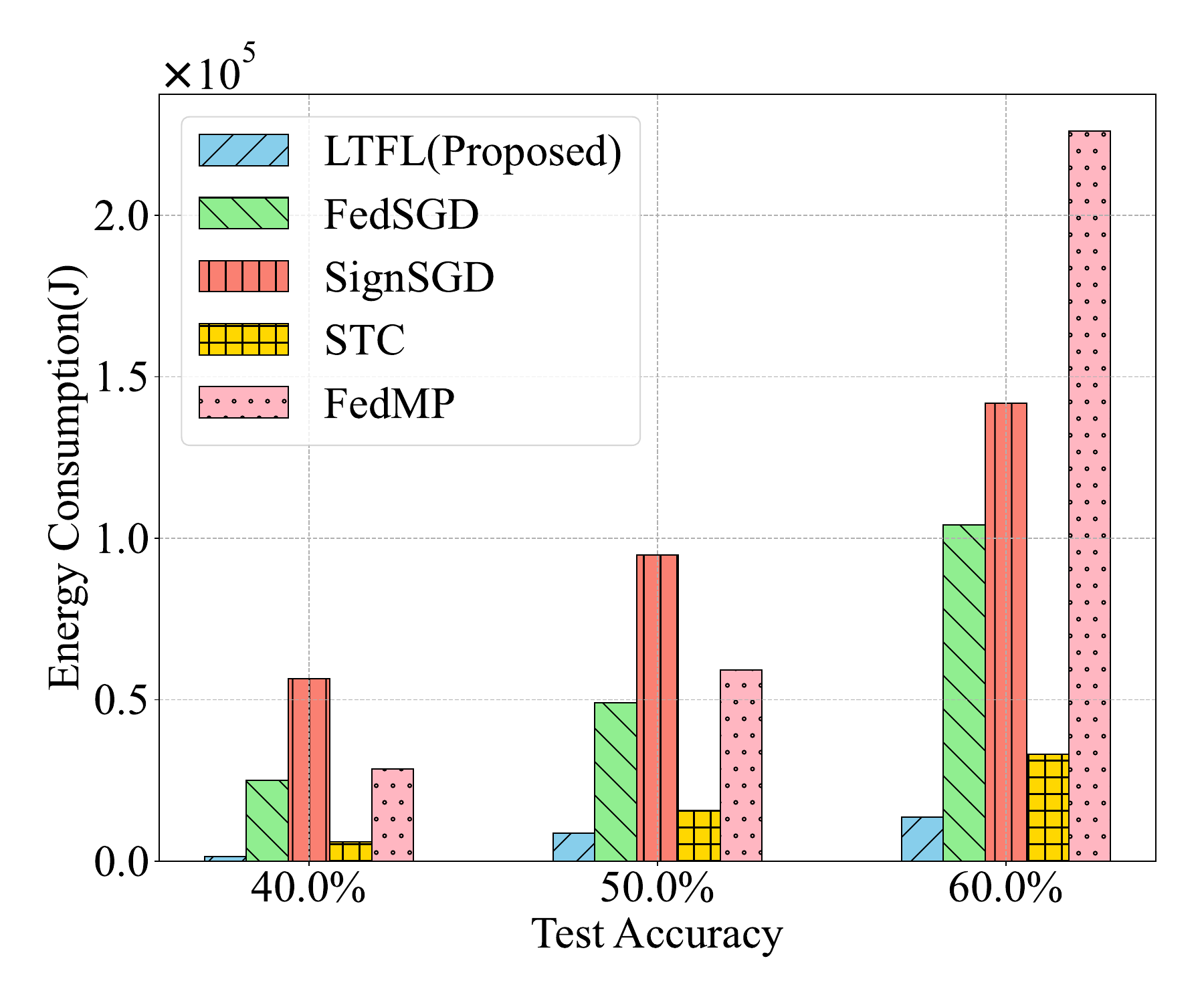}
  \end{minipage}
  }
  \subfigure[{\color{black}Energy consumption comparison (concentration parameter $\alpha=0.9$).}]{
  \begin{minipage}{0.31\linewidth}\label{ComparisonNoniid0.9_EnergyConsumption}
  \centering
  \includegraphics[height=4.3cm, width=5.8cm]{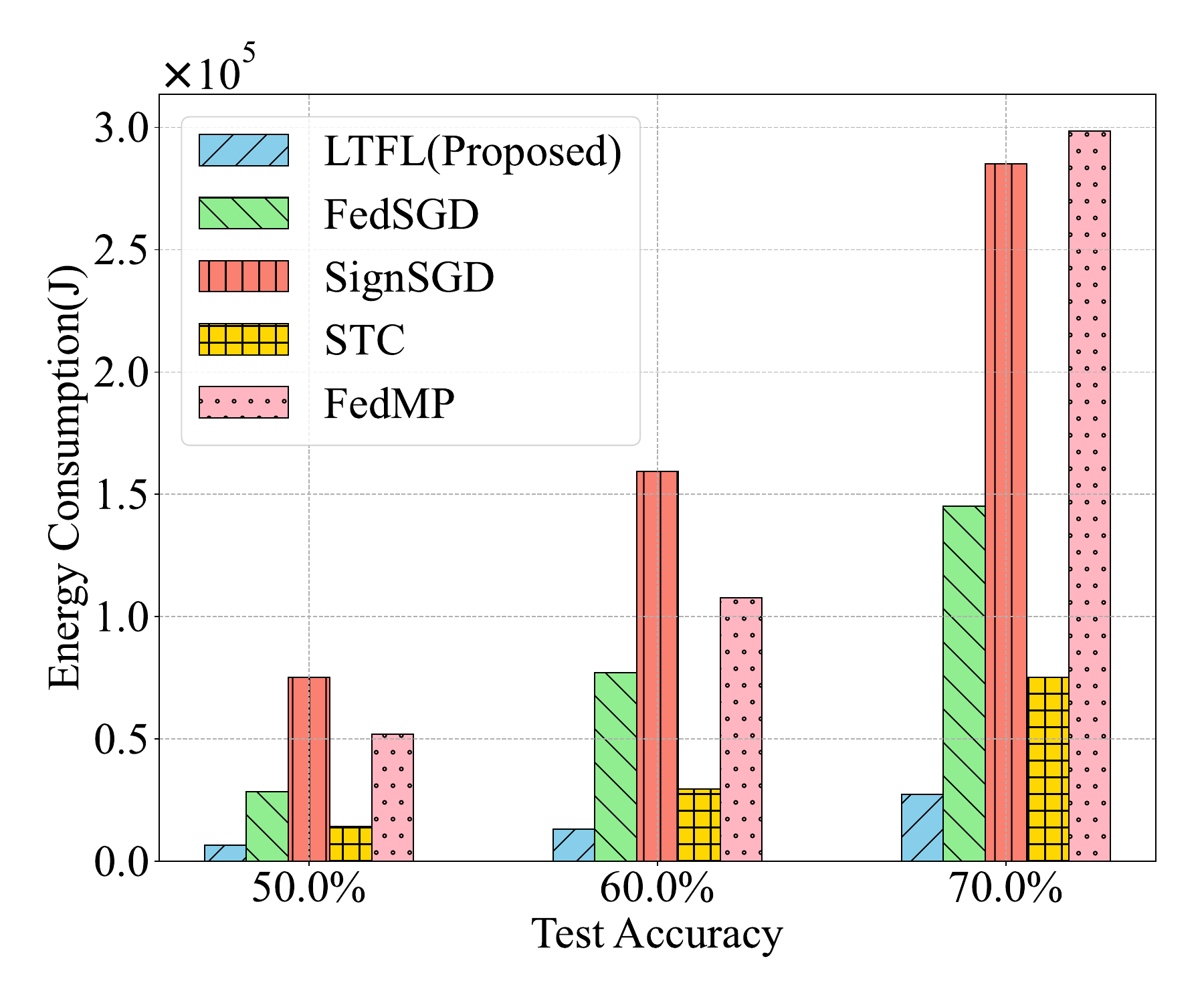}
  \end{minipage}
  }
  \centering
  \caption {{\color{black}Energy consumption comparison of different schemes under different non-i.i.d. scenarios.}}\label{ComparisonNoniid_EnergyConsumption}
\end{figure*}

\section{Conclusion}
{This paper presents a lightweight FL framework over wireless networks, incorporating model pruning, gradient quantization, and transmission power control. To minimize the training loss while satisfying delay and energy consumption constraints, we firstly established a closed-form expression for evaluating the impact of unreliable wireless transmission, model pruning error, and gradient quantization error on the FL's convergence. Furthermore, guided by the insightful theoretical results, we formulated a systematic optimization problem on jointly optimizing model pruning ratio, quantization level, and transmission power control. Then to solve the problem, a two-stage algorithm was proposed, with the first stage providing closed-form optimal solutions for model pruning ratio and gradient quantization level, and the second stage utilizing Bayesian optimization to determine the optimal transmission power. Experimental results on real-world data demonstrated its superiority over existing schemes.

Due to the high cost of experimentation, this study lacks validation in real-world environments. In future work, we aim to evaluate the proposed method in practical wireless network settings. Furthermore, our approach is particularly suitable for training large and complex models. As large models are expected to become indispensable tools in future human society, retraining them from scratch is often impractical, and task-specific fine-tuning has emerged as a promising alternative. Future research could explore how to adapt our method for efficient fine-tuning of large models, particularly in resource-constrained edge environments where communication and computation efficiency are critical.}
\bibliography{MPGDFLref}

\clearpage
\begin{appendices}
\setcounter{equation}{0}
\renewcommand\theequation{A.\arabic{equation}}
\section{Proof of Theorem \ref{Theorem1}} \label{ProofTheorem}
\begin{proof}
For simplicity, we use ${\boldsymbol{g}}^n({\boldsymbol{w}}^{n})$ to represent ${\boldsymbol{g}}^n({\boldsymbol{w}}^{n};\{p_u^n\})$; ${\boldsymbol{g}}^n(\{\hat{\boldsymbol{w}}_u^{n}\})$ to represent ${\boldsymbol{g}}^n\left(\{\hat{\boldsymbol{w}}_u^{n}\} ; \{\rho_u^n\}, \{p_u^n\} \right)$, and $\overline{\boldsymbol{g}}^n\left(\{\hat{\boldsymbol{w}}_u^{n}\} \right)$ represents $\overline{\boldsymbol{g}}^n\left(\{\hat{\boldsymbol{w}}_u^{n}\} ; \{\rho_u^n\}, \{\delta_u^n\}, \{p_u^n\}\right)$.

To facilitate the following analysis, we introduce an auxiliary variables as
 \begin{equation}\label{AuxiliraryVar1}
\boldsymbol{\lambda}_1^n=\nabla F\left(\boldsymbol{w}^n\right)-\overline{\boldsymbol{g}}^n(\{\hat{\boldsymbol{w}}_u^n\}),
\end{equation}
respectively. Hence, Eq. (\ref{ModelUpdateReality}) can be rewritten as
\begin{equation}
\label{UpdateWithAuxiliaryVariable}
\boldsymbol{w}^{n+1}=\boldsymbol{w}^n-\eta\left(\nabla F\left(\boldsymbol{w}^n\right)-\boldsymbol{\lambda}_1^n\right).
\end{equation}

Furthermore, we rewrite $F\left(\boldsymbol{w}^{n+1}\right)$ as the expression of its second-order Taylor expansion, which is given by
\begin{equation}\label{TaylorExpansion}
\begin{aligned}
F\left(\boldsymbol{w}^{n+1}\right)  \leq & F\left(\boldsymbol{w}^n\right)+\left(\nabla F\left(\boldsymbol{w}^n\right)\right)^{\top}\left(\boldsymbol{w}^{n+1}-\boldsymbol{w}^n\right)\\&+\frac{\nabla^2 F\left(\boldsymbol{w}^n\right)}{2}\left\|\boldsymbol{w}^{n+1}-\boldsymbol{w}^n\right\|^2 \\
  \stackrel{(a)}{\leq} & F\left(\boldsymbol{w}^n\right)+ \left(\nabla F\left(\boldsymbol{w}^n\right)\right)^{\top}\left(\boldsymbol{w}^{n+1}-\boldsymbol{w}^n\right)\\&+\frac{L}{2}\left\|\boldsymbol{w}^{n+1}-\boldsymbol{w}^n\right\|^2 \\
  \stackrel{(b)}{\leq}& F\left(\boldsymbol{w}^n\right)-\eta \left(\nabla F\left(\boldsymbol{w}^n\right)\right)^{\top} \left(\nabla F\left(\boldsymbol{w}^n\right)-\boldsymbol{\lambda}_1^n\right) \\& + \frac{L \eta^2}{2}\left\|\nabla F\left(\boldsymbol{w}^n\right)-\boldsymbol{\lambda}_1^n\right\|^2 ,\\
\end{aligned}
\end{equation}
where inequality (a) stems from  Eq. (\ref{Assumtion3}), and inequality (b) is due to Eq. (\ref{UpdateWithAuxiliaryVariable}). Given learning rate $\eta=\frac{1}{L}$,  we have
\begin{equation}\label{ReFWithGivenLearningRate}
\begin{aligned}
\mathbb{E}& \left\{F\left(\boldsymbol{w}^{n+1}\right)\right\}   \\ \leq &  \mathbb{E} \left\{ F\left(\boldsymbol{w}^n\right)-\frac{1}{L}\left\|\nabla F\left(\boldsymbol{w}^n\right)\right\|^2+ \frac{1}{2 L}\left\|\nabla F\left(\boldsymbol{w}^n\right)\right\|^2  \right. \\
&\left. +\frac{1}{2 L}\left\|\boldsymbol{\lambda}_1^n\right\|^2 +\frac{1}{L}\left(\boldsymbol{\lambda}_1^n\right)^{\top} \nabla F\left(\boldsymbol{w}^n\right)-\frac{1}{L}\left(\boldsymbol{\lambda}_1^n\right)^{\top} \nabla F\left(\boldsymbol{w}^n\right) \right\} \\
 \leq & \mathbb{E} \left\{ F\left(\boldsymbol{w}^n\right)-\frac{1}{2 L}\left\|\nabla F\left(\boldsymbol{w}^n\right)\right\|^2+\frac{1}{2 L}\left\|\boldsymbol{\lambda}_1^n\right\|^2 \right\}. \\
\end{aligned}
\end{equation}
Due to Eq. (\ref{AuxiliraryVar1}) , we have
\begin{equation}\label{ReAuxiliaryVar1}
\begin{aligned}
\mathbb{E} \left\{\left\|\boldsymbol{\lambda}_1^n\right\|^2\right\}=&\mathbb{E} \left\{\left\|\nabla F\left(\boldsymbol{w}^n\right)-\bar{\boldsymbol{g}}(\{\hat{\boldsymbol{w}}_u^{n}\})\right\|^2\right\} \\
 =& \mathbb{E} \left\{\left\|\nabla F\left(\boldsymbol{w}^n\right)-\boldsymbol{g}\left(\boldsymbol{w}^n\right)+\boldsymbol{g}\left(\boldsymbol{w}^n\right)-{\boldsymbol{g}}(\{\hat{\boldsymbol{w}}_u^{n}\}) \right.\right.\\&+\left. \left. {\boldsymbol{g}}(\{\hat{\boldsymbol{w}}_u^{n}\})-\bar{\boldsymbol{g}}(\{\hat{\boldsymbol{w}}_u^{n}\})\right\|^2 \right\}\\
 \stackrel{(\text {c})}{\leq} 
 & 3 \left(\mathbb{E} \left\{ \left\|\nabla F\left(\boldsymbol{w}^n\right)-\boldsymbol{g}\left(\boldsymbol{w}^n\right)\right\|^2 \right\} \right .\\&+ \mathbb{E} \left\{\left\|\boldsymbol{g}\left(\boldsymbol{w}^n\right)-{\boldsymbol{g}}(\{\hat{\boldsymbol{w}}_u^{n}\})\right\|^2 \right\} \\& + \left. \mathbb{E} \left\{\left\|{\boldsymbol{g}}(\{\hat{\boldsymbol{w}}_u^{n}\})-\bar{\boldsymbol{g}}(\{\hat{\boldsymbol{w}}_u^{n}\})\right\|^2 \right\}\right).
\end{aligned}
\end{equation}
where inequality (c) arises from Cauchy-Buniakowsky-Schwarz inequality (i.e., \\$\left\|\sum_{i=1}^n a_i b_i \right\|^2 {\leq}\sum^n_{i=1}\left\|a_i\right\|^2 \sum_{i=1}^n\left\|b_i\right\|^2  $).

In the following, we investigate the upper bounds of $\mathbb{E} \left\{ \left\|\nabla F\left(\boldsymbol{w}^n\right)-\boldsymbol{g}\left(\boldsymbol{w}^n\right)\right\|^2 \right\}$, $\mathbb{E} \left\{\left\|\boldsymbol{g}\left(\boldsymbol{w}^n\right)-{\boldsymbol{g}}(\{\hat{\boldsymbol{w}}^n_u\})\right\|^2 \right\}$, and$\mathbb{E} \left\{\left\|{\boldsymbol{g}}(\{\hat{\boldsymbol{w}}^n_u\})-\bar{\boldsymbol{g}}(\{\hat{\boldsymbol{w}}^n_u\})\right\|^2 \right\}$, respectively. Firstly,
\begin{equation}\label{UpperBoundofPower1}
\begin{aligned}
& \mathbb{E}\left\{\|\nabla F({\boldsymbol{w}}^n) - {\boldsymbol{g}}({\boldsymbol{w}}^n)\|^2\right\} 
\\& = \mathbb{E}\left\{\left\|\frac{\sum_{u=1}^{U} N_u  \nabla F_u({\boldsymbol{w}}^n)}{N} - \frac{\sum_{u=1}^{U} N_u \alpha_u^n  \nabla F_u({\boldsymbol{w}}^n)}{\sum_{u=1}^{U} N_u \alpha_u^n}\right\|^2\right\} 
\\& = \mathbb{E}\left\{\left\|\frac{\sum_{u \in \mathcal{U}_1} N_u  \nabla F_u({\boldsymbol{w}}^n) + \sum_{u \in \mathcal{U}_2} N_u  \nabla F_u({\boldsymbol{w}}^n)}{N} \right. \right. - 
\\ & \left. \left. \frac{\sum_{u=1}^{U} N_u \alpha_u^n  \nabla F_u({\boldsymbol{w}}^n)}{\sum_{u=1}^{U} N_u \alpha_u^n}\right\|^2\right\} 
\\& = \mathbb{E}\left\{\left\|\frac{(\sum_{u=1}^{U} N_u \alpha_u^n - N)(\sum_{u \in \mathcal{U}_1} N_u \nabla F_u({\boldsymbol{w}}^n))}{N \sum_{u=1}^{U} N_u \alpha_u^n} \right. \right. + 
\\ & \left. \left. \frac{\sum_{u \in \mathcal{U}_2} N_u \nabla F_u({\boldsymbol{w}}^n)}{N}\right\|^2\right\} 
\\& \stackrel{(\text {d})}{\leq}\mathbb{E}\left\{\frac{(N - \sum_{u=1}^{U} N_u \alpha_u^n )(\sum_{u \in \mathcal{U}_1} N_u \left\|\nabla F_u({\boldsymbol{w}}^n)\right\|)}{N \sum_{u=1}^{U} N_u \alpha_u^n} \right.  + 
\\ &  \left. \frac{\sum_{u \in \mathcal{U}_2} N_u \left\|\nabla F_u({\boldsymbol{w}}^n)\right\|}{N}\right\}^2
\end{aligned} 
\end{equation}
where $\mathcal{U}_1 = \{\alpha_u^n = 1 | u \in \mathcal{U}\}$ is the set of users that correctly transmit their local FL models to the BS and $\mathcal{U}_2 = \{u \in \mathcal{U} | u \notin \mathcal{U}_1\}$. The inequality equation
in (d) is achieved by the triangle-inequality. Since $\left\|\nabla f\left(\boldsymbol{w}^{n}; \boldsymbol{x}_{u,i}, \boldsymbol{y}_{u,i}\right)\right\| \leq \sqrt{\upsilon_{1}+\upsilon_{2}\left\|\nabla F\left(\boldsymbol{w}^{n}\right)\right\|^{2}}$, 
we have
\begin{equation}
    \sum_{u \in \mathcal{U}_1} N_u \left\|\nabla F_u({\boldsymbol{w}}^n)\right\| \leq \sqrt{\upsilon_{1}+\upsilon_{2}\|\nabla F(\boldsymbol{w}^{n})\|^{2}}\sum _{u=1}^U N_u\alpha_u^n
\end{equation}
and 
\begin{equation}
    \sum_{u \in \mathcal{U}_2} N_u \left\|\nabla F_u({\boldsymbol{w}}^n)\right\| \leq \sqrt{\upsilon_{1}+\upsilon_{2}\|\nabla F(\boldsymbol{w}^{n})\|^{2}}\left(N - \sum _{u=1}^U N_u\alpha_u^n\right).
\end{equation}
Thus, Eq. \ref{UpperBoundofPower1} can be simplified as follows 
\begin{equation}\label{UpperBoundOfPower2}
\begin{aligned}
& \mathbb{E}\left\{\|\nabla F({\boldsymbol{w}}^n) - {\boldsymbol{g}}({\boldsymbol{w}}^n)\|^2\right\} \\&  
\stackrel{\text {  }}{\leq} \frac{4}{N^{2}} \mathbb{E}\left\{N-\sum_{u=1}^{U} N_u \alpha_{u}^n\right\}^{2}\left(\upsilon_{1}+\upsilon_{2}\|\nabla F(\boldsymbol{w}^{n})\|^{2}\right) \\& \stackrel{\text { (e) }}\leq \frac{4}{N} \mathbb{E}\left\{N-\sum_{u=1}^{U} N_u \alpha_{u}^n\right\}\left(\upsilon_{1}+\upsilon_{2}\|\nabla F(\boldsymbol{w}^{n})\|^{2}\right) \\& = \frac{4}{N}\left(\upsilon_{1}+\upsilon_{2}\|\nabla F(\boldsymbol{w}^{n})\|^{2}\right) \sum_{u=1}^U N_u q_u^n,
\end{aligned}
\end{equation}
where inequality (e) is due to the fact that $N \geq N-\sum_{u=1}^{U} N_u \alpha_{u} \geq 0$.

Secondly, the upper bound of $\mathbb{E}\left(\|\boldsymbol{g}({\boldsymbol{w}}^n) - {\boldsymbol{g}}(\{\hat{\boldsymbol{w}}^n_u\})\|^2\right)$ can be represented as
\begin{equation}\label{UpperBoundOfPruning}
\begin{aligned}
& \mathbb{E} \left\{ \left\| {\boldsymbol{g}}(\boldsymbol{w}^n) - {\boldsymbol{g}}(\{\hat{\boldsymbol{w}}^n_u\}) \right\|^2 \right\}
\\&= \mathbb{E} \left\{ \left\| \frac{\sum_{u=1}^{U} N_u \alpha_u^n (\nabla F_u(\boldsymbol{w}^n) - \nabla F_u(\hat{\boldsymbol{w}}_u^n))}{\sum_{u=1}^{U}N_u\alpha_u^n} \right\|^2 \right\}
\\& \stackrel{(\text{f})}{\leq} \mathbb{E} \left\{ \frac{\left(\sum_{u=1}^{U} \|N_u\alpha_u^n\|^2\right) \sum_{u=1}^{U} \left\| \nabla F_u(\boldsymbol{w}^n) - \nabla F_u(\hat{\boldsymbol{w}}_u^n) \right\|^2}{\left\|\sum_{u=1}^{U}N_u\alpha_u^n\right\|^2} \right\} 
\\& \stackrel{(\text {g})}{\leq} \mathbb{E} \left\{ \frac{\sum_{u=1}^{U}\| N_u\alpha_u^n\|^2}{\left\|\sum_{u=1}^{U}N_u\alpha_u^n\right\|^2} \cdot \sum_{u=1}^{U} L^2 \left\| \boldsymbol{w}^n - \hat{\boldsymbol{w}}_u^n \right\|^2 \right\} 
\\& \stackrel{(\text {h})}{\leq} \mathbb{E} \left\{ \frac{\sum_{u=1}^{U}\| N_u\alpha_u^n\|^2}{\sum_{u=1}^{U}\left\|N_u\alpha_u^n\right\|^2} \cdot \sum_{u=1}^{U} L^2 \left\| \boldsymbol{w}^n - \hat{\boldsymbol{w}}_u^n \right\|^2 \right\} 
\\& \stackrel{(\text {i})}{\leq}  L^2 D^2   \cdot \sum_{u=1}^{U} \rho_u \triangleq \Gamma_1^n 
\end{aligned}
\end{equation}
where inequality (f) is because of Cauchy-Buniakowsky-Schwarz inequality, while inequality (g) is from the Assumption \ref{CompleteAssumption1}. Inequality (h) is because of the fact that $\sum_{i=1}^n a_i^2 \leq\left(\sum_{i=1}^n a_i\right)^2$ and inequality (i) stems from Lemma \ref{ModelPruningLemma}.
Thirdly, the upper bound of  $ \mathbb{E}\left\{\left\|{\boldsymbol{g}}(\{\hat{\boldsymbol{w}}^n_u\})-\overline{\boldsymbol{g}}(\{\hat{\boldsymbol{w}}^n_u\})\right\|^2\right\}$ can be represented as 
\begin{equation}\label{UpperBoundOfQuantization}\small
\begin{aligned}
& \mathbb{E}\left\{\left\|{\boldsymbol{g}}(\{\hat{\boldsymbol{w}}^n_u\})-\overline{\boldsymbol{g}}(\{\hat{\boldsymbol{w}}^n_u\})\right\|^2\right\}\\
& =\mathbb{E}\left\{\left\|\frac{\sum_{u=1}^U N_u \alpha_u^n \left({\boldsymbol{g}}_u\left(\hat{\boldsymbol{w}}_u^n\right) - \mathcal{Q}\left({\boldsymbol{g}}_u\left(\hat{\boldsymbol{w}}_u^n\right)\right)\right)}{\sum_{u=1}^U N_u \alpha_u^n}\right\|^2\right\} \\
& \stackrel{\text { (j) }}{\leq} \mathbb{E}\!\left\{\!\frac{\left(\sum\limits_{\!u=1}^{\! U}\!\left\|N_u \alpha_u^n \right\|^2 \right)\!\left(\sum\limits_{u=1}^U \! \left\|{\boldsymbol{g}}_u\left(\hat{\boldsymbol{w}}_u^n\right) - \mathcal{Q}\left({\boldsymbol{g}}_u\left(\hat{\boldsymbol{w}}_u^n\right)\right)\right\|^2\!\right)}{\left\|\sum_{u=1}^U N_u \alpha_u^n \right\|^2}\!\right\}\\
& \stackrel{\text {  }}{\leq}  \sum _ { u = 1}^{U} \mathbb{E}\left\{\left\|{\boldsymbol{g}}_u\left(\hat{\boldsymbol{w}}_u^n\right) - \mathcal{Q}\left({\boldsymbol{g}}_u\left(\hat{\boldsymbol{w}}_u^n\right)\right)\right\|^2\right\}\\
& \stackrel{\text { (k) }}{\leq}   \sum _ { u = 1 }^{U}\! \frac{\sum_{v=1}^V\left(\bar{g}_{u, v}^n-\underline{g}_{u, v}^n\right)^2}{4\left(2^{\delta_{u}^n}-1\right)^2} \!\triangleq \Gamma_2^n,
\end{aligned}
\end{equation}
where inequality (j) is due to Cauchy-Buniakowsky-Schwarz inequality, while inequality (k) stems from Lemma \ref{QuantizationLemma}. 

Therefore, substituting Eq. (\ref{UpperBoundOfPruning}), Eq. (\ref{UpperBoundOfPower2}), and Eq. (\ref{UpperBoundOfQuantization}) into Eq. (\ref{ReAuxiliaryVar1}), we can obtain
\begin{equation}\label{UpperBoundOfAuxiliaryVar1}
\begin{aligned}
\mathbb{E}&\left\{ \left\|\boldsymbol{\lambda}_1^n\right\|^2 \right\} \leq3 \Gamma_2^n  +3\Gamma_1^n \\ &+ \frac{12}{N}\left(\upsilon_{1}+\upsilon_{2}\|\nabla F(\boldsymbol{w}^{n})\|^{2}\right) \sum_{u=1}^U N_u q_u^n.
\end{aligned}
\end{equation}
Furthermore, let we substitute Eq. (\ref{UpperBoundOfAuxiliaryVar1}) into Eq. (\ref{ReFWithGivenLearningRate}), we have
\begin{equation}\label{ReReFWithGivenLearningRate}
\begin{aligned}
& \mathbb{E}\left\{F(\boldsymbol{w}^{n+1})\right\} \leq \mathbb{E}\left\{F(\boldsymbol{w}^n)\right\} + \frac{6\upsilon_1}{LN} \sum_{u=1}^U N_u q_u^n - 
\\& \frac{1}{2L}(1-\frac{12\upsilon_2}{N} \cdot \sum_{u=1}^U N_u q_u^n) \|\nabla F(\boldsymbol{w}^n)\|^2 + \frac{3}{2L}(\Gamma_1^n+\Gamma_2^n)\\\end{aligned}
\end{equation}
Rearranging Eq. (\ref{ReReFWithGivenLearningRate}), we can obtain
\begin{equation}
\begin{aligned}
C^n \mathbb{E}\left\{\left\|\nabla F\left({\boldsymbol{w}}^n\right)\right\|^2 \right\} \leq  & 2 L\mathbb{E}\left\{F\left(\boldsymbol{w}^n\right)-F\left(\boldsymbol{w}^{n+1}\right)\right\} \\& + \frac{12\upsilon_1}{N} \sum_{u=1}^U N_u q_u^n+3\Gamma_1^n+3\Gamma_2^n,
\end{aligned}
\end{equation}
where $C^n = 1 - \frac{12\upsilon_2}{N}\sum_{u=1}^UN_uq_u^n \in (1-12\upsilon_2,1)$. 
Summing up the above terms from $n = 0$ to $\Omega$ and dividing both sides by the total number of iterations, we can obtain
\begin{equation}\label{ConvergenceRate}
\begin{aligned}
&\frac{1-12\upsilon_2}{\Omega+1}\sum_{n=0}^{\Omega}\mathbb{E}\left\{\left\|\nabla F\left({\boldsymbol{w}}^n\right)\right\|^2\right\} \\&\leq \frac{2L}{\Omega+1}\mathbb{E}\left\{F\left(\boldsymbol{w}^0\right)-F\left(\boldsymbol{w}^{*}\right)\right\}\\&+\sum_{n=0}^{\Omega}\left(\frac{12\upsilon_1}{N} \sum_{u=1}^U N_u q_u^n\right)  +\frac{3}{\Omega+1}\sum_{n=0}^{\Omega}\Gamma_1^n+\frac{3}{\Omega+1} \sum_{n=0}^{\Omega}\Gamma_2^n.
\end{aligned}
\end{equation}
where $\boldsymbol{w}^{*}$ is the optimal model.
Let $\Gamma^n= \frac {1}{1-12\upsilon_2}\left(3 \Gamma_2^n+3\Gamma_1^n+\frac{12\upsilon_1}{N} \sum_{u=1}^U N_u q_u^n\right)$, and then Eq. (\ref{ConvergenceRate}) can be rewritten as
\begin{equation}\label{ReConvergenceRate}
\begin{aligned}
 &\frac{1}{\Omega+1}\sum_{n=0}^{\Omega}\mathbb{E}\left\{\left\|\nabla F\left({\boldsymbol{w}}^n\right)\right\|^2\right\} \\& \leq \frac{2L}{(1-12\upsilon_2)(\Omega+1)}\mathbb{E}\left\{F\left(\boldsymbol{w}^0\right)-F\left(\boldsymbol{w}^{*}\right)\right\}\\& + \frac{1}{(\Omega+1)} \sum_{n=0}^{\Omega} \Gamma^n.
\end{aligned}
\end{equation}
This completes the proof.
\end{proof}

\setcounter{equation}{0}
\renewcommand\theequation{B.\arabic{equation}}
\section{Proof of Theorem \ref{theorem2}} \label{ProofTheorem2}
Evidently, the objective function shown in Eq. (\ref{OptimizationProblem}) is monotonically increasing, which indicates that $\rho_u^n,\forall u, \forall n,$ should be the minimum to minimize the convergence error gap. The minimum of  $\rho_u^n, \forall u, \forall n,$  is mainly determined by constraints Eq. (\ref{Tmax}) and Eq. (\ref{Emax}). We expand the constraint in Eq. (\ref{Tmax}) as
\begin{equation}
\label{TmaxSubP1_Exp}
\frac{{N}_{u}c_0\left(1-\rho_u^n\right)}{f_u^n}+\frac{\tilde{\delta}_u^n\left(1-\rho_u^n\right)}{R_u^{\mathrm{n}}\left( p^{n}_u\right) }  \leq T^{\rm max}-s.
\end{equation}
Since ${N}_{u}$, $ c_0 $, $f_u^n$, $\tilde{\delta}_u^n$, and $R_u^{\mathrm{n}}\left( p^{n}_u\right) $ are positive, Eq. (\ref{TmaxSubP1_Exp}) can be rewritten as
\begin{equation}
\label{TmaxSubP1_Exp2}
\rho_u^n \geq 1- \frac{T^{\rm max}-s }{\frac{{N}_{u}c_0}{f_u^n}+\frac{\tilde{\delta}_u^n}{R_u^{\mathrm{n}}\left( p^{n}_u\right) }}.
\end{equation}
Furthermore, we expand constraint in Eq. (\ref{Emax}) as
\begin{equation}
\label{EmaxSubP1_Exp}
k(f_u^n)^{\sigma-1}{N}_{u}c_0\left(1-\rho_u^n\right)+\frac{p_u^n\tilde{\delta}_u^n\left(1-\rho_u^n\right)}{R_u^{\mathrm{n}}\left( p^{n}_u\right)}\leq E^{\rm max}.
\end{equation}
Similarly, Eq. (\ref{EmaxSubP1_Exp}) can be rewritten as
\begin{equation}
\label{EmaxSubP1_Exp2}
\rho_u^n \geq 1- \frac{E^{\rm max}} {k(f_u^n)^{\sigma-1}{N}_{u}c_0+p_u^n\frac{\tilde{\delta}_u^n}{R_u^{\mathrm{n}}\left( p^{n}_u\right)}}.
\end{equation}
Therefore, according to Eq. (\ref{TmaxSubP1_Exp2}),  Eq. (\ref{EmaxSubP1_Exp2}), and Eq. (\ref{ModelPruningRatioCons}),  $(\rho_u^{n})^*, \forall u, \forall n$ should satisfy
\begin{equation}
(\rho_u^{n})^*= \min\left\{\rho^{\rm max}, \left(1- \min\left\{ \Phi_1,\Phi_2     \right\} \right)^{+}\right\},
\end{equation}
where $(x)^{+}=\max\left\{0,x\right\}$, and  $\Phi_1$ and $\Phi_2$ can be given by
\begin{equation}
\Phi_1= \frac{T^{\rm max}-s}{\frac{{N}_{u}c_0}{f_u^n}+\frac{\tilde{\delta}^n_u}{R_u^{\mathrm{n}}\left( p^{n}_u\right)}} ,
\end{equation}
and
\begin{equation}
\Phi_2= \frac{E^{\rm max}}{k(f_u^n)^{\sigma-1}{N}_{u}c_0+\frac{p_u^n\tilde{\delta}_u^n}{R_u^{\mathrm{n}}\left( p^{n}_u\right)}},
\end{equation}
respectively.

\setcounter{equation}{0}
\renewcommand\theequation{C.\arabic{equation}}
\section{Proof of Lemma \ref{Lemma3}} \label{Lemma3Proof}
\begin{proof}
Define
\begin{equation}
\begin{aligned}
f_1(\tilde{\delta}_u^n) =   \frac{\sum_{v=1}^V\left(\overline{g}_{u, v}^n-\underline{g}_{u, v}^n\right)^2}{4\left(2^{\tilde{\delta}_u^n}-1\right)^2}.
\end{aligned}
\end{equation}
The first term of the objective function of problem $\mathcal{P}3$ is  the sum of  $f_1(\tilde{\delta}_u^n)$  from $u=1$ to $u=U$, and   $f_1(\tilde{\delta}_{u}^n)$ and $f_1(\tilde{\delta}_{\dot{u}}^n)$ (${\dot{u}}\neq{{u}}$) are independent. Therefore, we can obtain the monotonicity of the main term of the objective function by discussing $f_1(\tilde{\delta}_u^n)$.
The first derivative of $f_1(\tilde{\delta}_u^n) $ with respective to $\tilde{\delta}_u^n$ is given by
\begin{equation}
\begin{aligned}
\dot{f}_1(\tilde{\delta}_u^n) &= -\frac{2\ln2\sum_{v=1}^V(\overline{g}_{u, v}^n-\underline{g}_{u, v}^n)^2}{(2^{\delta_u^n}-1)^2}.
\end{aligned}
\end{equation}
Since $\sum_{v=1}^V(\overline{g}_{u, v}^n-\underline{g}_{u, v}^n)^2$ and ${(2^{\delta_u^n}-1)^2}$ are always positive, we have
\begin{equation}
\begin{aligned}
\dot{f}_1(\tilde{\delta}_u^n) & \leq 0,
\end{aligned}
\end{equation}
which indicates that  $f_1(\tilde{\delta}_u^n)$ is  monotonically decreasing. Since the first term of the objective function of  problem $\mathcal{P}3$ is the sum of  $f_1(\tilde{\delta}_u^n)$  from $u=1$ to $u=U$, and other terms are constants, problem $\mathcal{P}3$ is monotonically decreasing with respect to $\boldsymbol{\delta}^{n}$.
\end{proof}

\setcounter{equation}{0}
\renewcommand\theequation{D.\arabic{equation}}
\section{Proof of Theorem \ref{theorem3}} \label{Theorem3Proof}
\begin{proof}
Due to the objective function of problem $\mathcal{P}3$ is monotonically decreasing,  $\delta_u^n,\forall u, \forall n,$ should be the maximum to minimize the convergence error gap. The maximum of  $\delta_u^n, \forall u, \forall n,$  is determined by the constraints in Eq. (\ref{TmaxSub2}) and Eq. (\ref{EmaxSub2}). Since $\left(1-  (\rho_u^n)^*\right)$ and ${B_u^{\rm UL}\log_2\left(1+\frac{p_u^nh_u^n}{I_u^n+B_u^{\rm UL}N_0}\right) } $ are positive, Eq. (\ref{TmaxSub2}) can be rewritten as
\begin{equation}
\begin{aligned}
\label{TmaxSub2ExpanisionRe}
\tilde{\delta}_u^n   \leq  \frac{\left(T^{\rm max}-s-\frac{{N}_{u}c_0\left(1-  (\rho_u^n)^*\right)}{f_u^n}\right)R_u^{\mathrm{n}}\left( p^{n}_u\right) }{1-  (\rho_u^n)^*}.
\end{aligned}
\end{equation}
Similarly, due to $p_u^n \geq 0$,  Eq. (\ref{EmaxSub2}) can be rewritten as
\begin{equation}
\label{EmaxSub2ExpanisionRe}
\tilde{\delta}_u^n\leq\frac{ \left(E_u^{\rm max}-k(f_u^n)^{\sigma-1}{N}_{u}c_0\left(1-(\rho_u^n)^*\right)\right)R_u^{\mathrm{n}}\left( p^{n}_u\right)}{p_u^n\left(1-(\rho_u^n)^*\right)} .
\end{equation}
We define
\begin{equation}
\begin{aligned}
\label{TmaxSub2ExpanisionRePhi}
\Phi_3 = \frac{\left(T^{\rm max}-s-\frac{{N}_{u}c_0\left(1-  (\rho_u^n)^*\right)}{f_u^n}\right)R_u^{\mathrm{n}}\left( p^{n}_u\right) }{1-  (\rho_u^n)^*}
\end{aligned}
\end{equation}
and
\begin{equation}
\label{EmaxSub2ExpanisionRePhi}
\Phi_4 = \frac{ \left(E_u^{\rm max}-k(f_u^n)^{\sigma-1}{N}_{u}c_0\left(1-(\rho_u^n)^*\right)\right)R_u^{\mathrm{n}}\left( p^{n}_u\right)}{p_u^n\left(1-(\rho_u^n)^*\right)}.
\end{equation}
Considering the relationship between $\tilde{\delta}_u^n$ and $\delta_u^n$ referring to Eq. (\ref{QuantizationRelationship}),  and the constraint in Eq. (\ref{QuantizationLevel_upperbound}), the optimal quantization strategy $(\delta_u^n)^*, \forall u, \forall n$ is given by
\begin{equation}
\label{OptimalQuantizationStrategyProof}
(\delta_u^n)^* =\left\lceil \min\left\{ \frac{\Phi_3-\xi}{V},\frac{\Phi_4-\xi}{V} , \delta^{\rm max} \right\}\right\rceil,
\end{equation}
where  $\lceil x \rceil$  represents the minimum positive integer that is less than or equal to $x$.
\end{proof}

\end{appendices}

\end{sloppypar}
\end{document}